\pdfoutput=1
\documentclass[11pt]{article}
\usepackage{amsmath}
\usepackage{amssymb}
\usepackage{amsthm}
\usepackage{algorithm}
\usepackage{algorithmicx}
\usepackage{algpseudocode}
\usepackage{caption}
\usepackage{diagbox}
\usepackage{footnote}
\usepackage{graphicx}
\usepackage{multirow}
\usepackage{subfigure}
\usepackage[T1]{fontenc}
\usepackage[margin=1 in]{geometry}

\newtheorem{theorem}{Theorem}
\newtheorem{lemma}{Lemma}

\newcounter{constraint}
\title{Scheduling Coflows for Minimizing the Makespan in Identical Parallel Networks}

\author{Chi-Yeh~Chen and Jun Chen 
\\ Department of Computer Science and Information
Engineering, \\ National Cheng Kung University, \\
Taiwan, ROC. \\
chency@csie.ncku.edu.tw, joe210150@gmail.com.}

\begin{document}

\maketitle
\begin{abstract}
With the rapid advancement of technology, parallel computing applications have become increasingly popular and are commonly executed in large data centers. These applications involve two phases: computation and communication, which are executed repeatedly to complete the work. However, due to the ever-increasing demand for computing power, large data centers are struggling to meet the massive communication demands. To address this problem, coflow has been proposed as a networking abstraction that captures communication patterns in data-parallel computing frameworks. This paper focuses on the coflow scheduling problem in identical parallel networks, where the primary objective is to minimize the makespan, which is the maximum completion time of coflows. It is considered one of the most significant $\mathcal{NP}$-hard problems in large data centers. In this paper, we consider two problems: flow-level scheduling and coflow-level scheduling. In the flow-level scheduling problem, distinct flows can be transferred through different network cores, whereas in the coflow-level scheduling problem, all flows must be transferred through the same network core. To address the flow-level scheduling problem, this paper proposes two algorithms: a $(3-\tfrac{2}{m})$-approximation algorithm and a $(\tfrac{8}{3}-\tfrac{2}{3m})$-approximation algorithm, where $m$ represents the number of network cores. For the coflow-level scheduling problem, this paper proposes a $(2m)$-approximation algorithm. Finally, we conduct simulations on our proposed algorithm and Weaver's algorithm, as presented in Huang \textit{et al.} (2020) in the 2020 IEEE International Parallel and Distributed Processing Symposium (IPDPS). We also validate the effectiveness of the proposed algorithms on heterogeneous parallel networks.
\end{abstract}

\section{Introduction}\label{section:introduction}
In recent years, the rapid growth in data volumes and the rise of cloud computing have revolutionized software systems and infrastructure. Numerous applications now dealing with large datasets sourced from diverse origins, presenting a formidable challenge in terms of efficient and prompt data handling. Consequently, the utilization of parallel computing applications has gained significant traction in large-scale data centers, as a means to tackle this pressing concern.

Data-parallel computation frameworks, including widely used ones like MapReduce~\cite{MapReduce}, Hadoop~\cite{Hadoop}, and Dyrad~\cite{Dryad}, offer the flexibility for applications to seamlessly transition between computation and communication stages. During the computation stage, intermediate data is generated and exchanged between sets of servers via the network. Subsequently, the communication stage involves the transfer of a substantial collection of flows, and the computation stage can only commence once all flows from the previous communication stage have been completed. Nevertheless, traditional networking approaches primarily prioritize optimizing flow-level performance rather than considering application-level performance metrics~\cite{Qiu}. Notably, in application-level performance metrics, the completion time of a job is determined solely by the last flow to finish the communication phase, disregarding any flows that may have completed earlier within the same stage. To tackle this issue, Chowdhury and Stoica~\cite{Coflow_cluster} introduced the concept of coflow abstraction, which takes into account application-level communication patterns for more comprehensive management and optimization.

A coflow, as defined by Qiu et al.~\cite{Qiu}, represents a collection of parallel flows that share a common performance goal. The data center is modeled as a non-blocking switch, depicted in Figure \ref{fig:non-blocking switch}, consisting of $N$ input and output ports, with the switch serving as a network core. The input ports facilitate data transfer from source servers to the network, while the output ports transfer data from the network to destination servers. Each coflow can be represented by an $N \times N$ demand matrix, where each element $d_{i,j}$ denotes the data volume transferred from input $i$ to output $j$ for the corresponding flow $(i,j)$. The capacity of all network links is assumed to be uniform, and capacity constraints apply to both input and output ports. The paper focuses on the coflow scheduling problem in identical parallel networks, aiming to minimize the makespan, which refers to the maximum completion time of all coflows. This problem is recognized as one of the most significant $\mathcal{NP}$-hard problems encountered in large-scale data centers. To address this challenge, the paper introduces approximation algorithms for both flow-level scheduling and coflow-level scheduling problems and evaluates their performance against existing algorithms through simulations.

\begin{figure}[!h]
\centering
\includegraphics[scale=0.2]{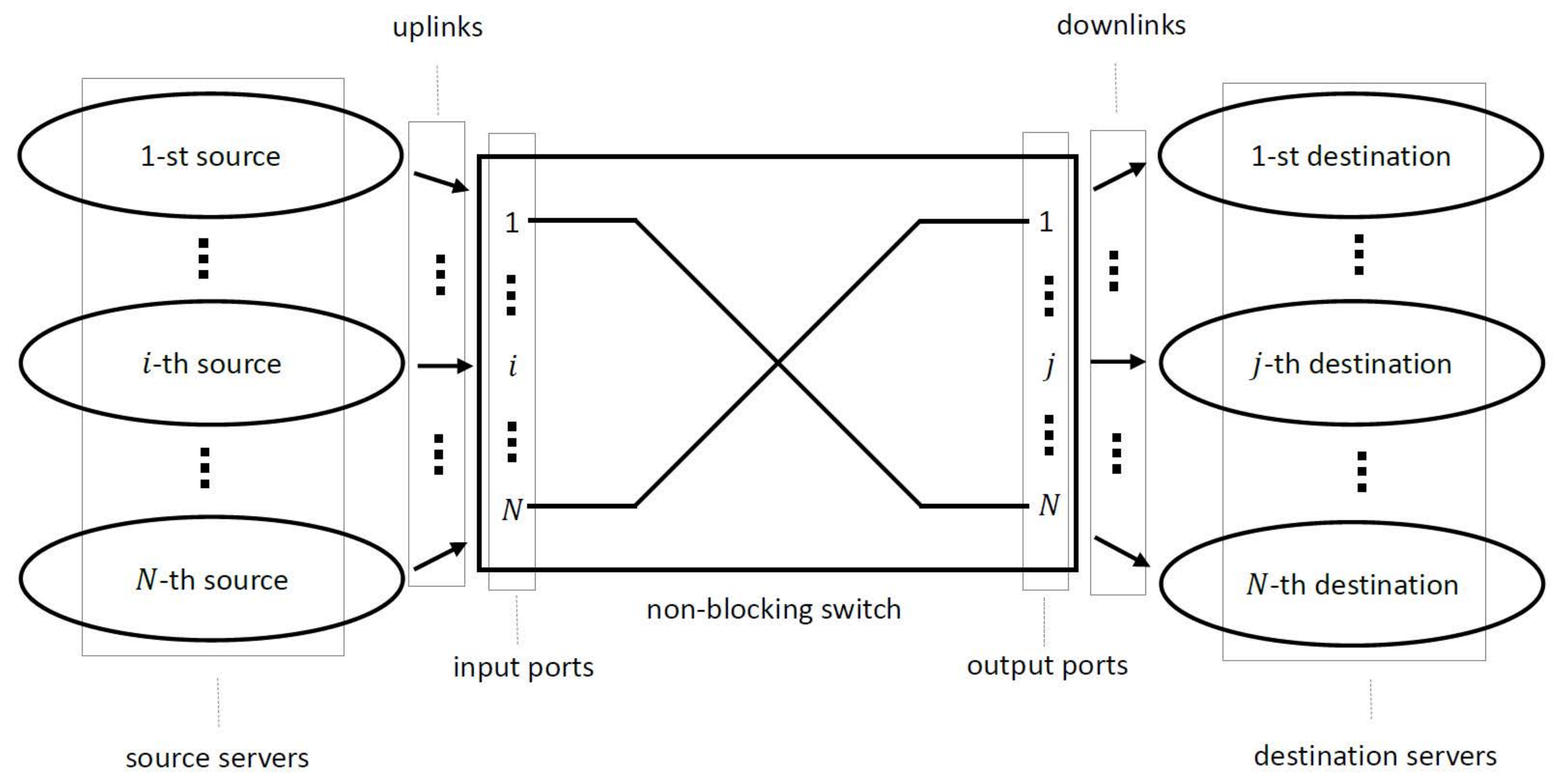}
\caption{A giant $N \times N$ non-blocking switch (network core).}
\label{fig:non-blocking switch}
\end{figure}

Previous studies on coflow scheduling problems~\cite{Qiu, Shafiee2, Shafiee} have predominantly focused on a single-core model. This model has been considered practical due to its utilization of topological designs such as Fat-tree or Clos networks~\cite{al2008scalable, greenberg2009vl2}, which facilitate the construction of data center networks with full bisection bandwidth. However, as technology trends evolve and computation networks become more intricate, the single-core model becomes insufficient in meeting the evolving requirements. It has been observed that modern data centers often employ multiple generations of networks simultaneously~\cite{singh2015jupiter} to bridge the gap in network speeds. As a result, our focus shifts towards identical parallel networks, where coflows can be transmitted through multiple identical network cores and processed in parallel. In this paper, the completion time of a coflow is defined as the time taken for the last flow within the coflow to complete. The objective of this paper is to schedule coflows in identical parallel networks, aiming to minimize the makespan, which represents the maximum completion time among all coflows.

This paper addresses the concept of coflow, which encompasses two distinct scheduling problems: flow-level scheduling and coflow-level scheduling. In the flow-level scheduling problem, individual flows can be distributed among multiple network cores for transmission. In contrast, the coflow-level scheduling problem limits the transmission of its constituent flows to a single network core. These two problems, representing different levels of granularity, capture the scheduling complexities associated with coflows.


\subsection{Our Contributions}
This paper addresses the coflow scheduling problem in identical parallel networks with the objective of minimizing the makespan. In the flow-level scheduling problem, we propose a $(3-\frac{2}{m})$-approximation algorithm as well as a $(\frac{8}{3}-\frac{2}{3m})$-approximation algorithm, where $m$ represents the number of network cores. Additionally, for the coflow-level scheduling problem, we present a $(2m)$-approximation algorithm.

\subsection{Organization}
We structure the remaining sections as follows. Section \ref{section:related-work} provides an overview of several related works. In Section \ref{section:notation-and-preliminaries}, we present the fundamental notations and preliminaries used in this paper. Our main results are presented in Section \ref{section:approximation-algorithm-for-divisible-coflow-scheduling} and Section \ref{section:approximation-algorithm-for-indivisible-coflow-scheduling}, where we provide two approximation algorithms for flow-level scheduling in Section \ref{section:approximation-algorithm-for-divisible-coflow-scheduling}, and one approximation algorithm for coflow-level scheduling in Section \ref{section:approximation-algorithm-for-indivisible-coflow-scheduling}. Subsequently, in Section \ref{section:experiments}, we conduct experiments to evaluate and compare the performance of our proposed algorithms with Weaver's \cite{Weaver}. Finally, in Section \ref{section:conclusion}, we present our conclusions.

\section{Related Work}\label{section:related-work}
In the literature, numerous heuristic algorithms have been proposed to tackle coflow scheduling problems, such as those discussed in \cite{Chowdhury, Varys, Hasnain, Shen}. Mosharaf \textit{et al.}~\cite{Varys} introduced a Smallest-Effective-Bottleneck-First heuristic that allocates coflows greedily based on the maximum server loads. They then utilized the Minimum-Allocation-for-Desired-Duration algorithm to assign rates to the corresponding flows. Dian \textit{et al.} \cite{Shen} conducted simulations to address the joint problem of coflow scheduling and virtual machine placement. They proposed a heuristic approach aimed at minimizing the completion time of individual coflows. Additionally, Chowdhury \textit{et al.} \cite{Chowdhury} presented a scheduler called Coflow-Aware Least-Attained Service, which operates without prior knowledge of coflows.

The concurrent open shop problem has been proven to be $\mathcal{NP}$-complete to approximate within a factor of $2-\epsilon$ for any $\epsilon>0$, in the absence of release times, as demonstrated in \cite{Sachdeva, Shafiee}. Interestingly, it is worth noting that each concurrent open shop problem can be reduced to a coflow scheduling problem. Consequently, the coflow scheduling problem is also NP-complete to approximate within a factor of $2-\epsilon$ for any $\epsilon>0$, when release times are not considered~\cite{Ahmadi, Shafiee}.

In the context of minimizing the total weighted completion time in identical parallel networks, Chen \cite{CYChen} proposed several approximation algorithms for the coflow scheduling problem under different conditions. For the flow-level scheduling problem, Chen devised an algorithm that achieved a $(6-\frac{2}{m})$-approximation with release time and a $(5-\frac{2}{m})$-approximation without release time, where $m$ represents the number of network cores. This algorithm employed an iterative approach to schedule all flows based on the completion time order of coflows, computed using a linear program. Subsequently, the algorithm assigned each flow to the least loaded network core to minimize the flow's completion time.

Regarding the coflow-level scheduling problem, Chen developed an algorithm that achieved a $(4m+1)$-approximation with release time and a $(4m)$-approximation without release time. Similar to the approach for divisible coflows, this algorithm also employed an iterative strategy to schedule flows based on the completion time order of coflows, computed using a linear program. Subsequently, the algorithm assigned each coflow to the least loaded network core to minimize the coflow's completion time.

In the context of coflow scheduling problems with precedence constraints, Chen \cite{CYChen_prec} also proposed two approximation algorithms for the aforementioned four conditions. In these cases, the approximation ratio for each condition, considering precedence constraints, is equal to the approximation ratio for each condition without precedence constraints, multiplied by a factor of $\mu$. Here, $\mu$ represents the coflow number of the longest path in the precedence graph.

In the domain of coflow scheduling problems for minimizing the total weighted completion time in a single network core, several related works have been proposed~\cite{Qiu, Shafiee2, Shafiee}. Qiu \textit{et al.} \cite{Qiu} introduced deterministic approximation algorithms achieving a $\frac{67}{3}$-approximation and a $\frac{64}{3}$-approximation with release time and without release time, respectively. Additionally, they obtained randomized approximation algorithms resulting in a $(9+\frac{16\sqrt{2}}{3})$-approximation and an $(8+\frac{16\sqrt{2}}{3})$-approximation with release time and without release time, respectively. The deterministic and randomized algorithms share a similar framework. They both approached the problem by relaxing it to a polynomial-sized interval-indexed linear program (LP), which provided an ordered list of coflows. Subsequently, the coflows were grouped based on their minimum required completion times from the ordered list. The algorithms scheduled coflows within the same time interval as a single coflow, utilizing matchings obtained through the Birkhoff-von Neumann decomposition theorem. The distinguishing factor between the deterministic and randomized algorithms lies in the selection of the time interval. The deterministic algorithm employed a fixed time point, while the randomized algorithm opted for a random time point.
However, Ahmadi \textit{et al.} \cite{Ahmadi} discovered that their approaches only yielded a deterministic $\frac{76}{3}$-approximation algorithm with release time. On the other hand, Shafiee \textit{et al.} \cite{Shafiee} presented the best-known results in recent work, proposing a deterministic $5$-approximation algorithm with release time and a $4$-approximation algorithm without release time. The deterministic algorithm employed a straightforward list scheduling strategy based on the order of the coflows' completion time, computed using a relaxed linear program that utilized ordering variables.

When addressing the coflow scheduling problem with the objective of minimizing the makespan in heterogeneous parallel networks, Huang \textit{et al.} \cite{Weaver} proposed an $O(m)$-approximation algorithm called Weaver, where $m$ represents the number of network cores. The Weaver algorithm scheduled all flows iteratively based on their size in descending order. Subsequently, it classified the flows into two categories: critical and non-critical. For critical flows, the algorithm selected a network that minimizes the coflow completion time. On the other hand, for non-critical flows, it selected a network to balance the load.

Furthermore, Chen \cite{CYChen_het} further improved upon the previous results by introducing an $O(\frac{\log m}{\log \log m})$-approximation algorithm. As a preprocessing step, Chen modified the makepan scheduling problem instance to contain only a small number of groups. In the first stage of preprocessing, any network cores that were at most $\frac{1}{m}$ times the speed of the fastest network core, where $m$ is the number of network cores, were discarded. The second stage of preprocessing involved dividing the remaining network cores into groups based on their similar speeds. Following the preprocessing step, Chen implemented the list algorithm to identify the least loaded network core and assigned flows to it. Additionally, Chen obtained an $O(\frac{\log m}{\log \log m})$-approximation algorithm for minimizing the total weighted completion time.

\section{Notation and Preliminaries}\label{section:notation-and-preliminaries}
Our work focuses on the abstraction of identical parallel networks, which is an architecture consisting of identical network cores operating in parallel. We consider the identical parallel networks as a set $\mathcal{M}$ of $m$ giant $N \times N$ non-blocking switches, where each switch represents a network core. These switches have $N$ input ports and $N$ output ports. The input ports are responsible for transferring data from source servers to the network, while the output ports transfer data from the network to destination servers. In the network, there are $N$ source servers, where the $i$-th source server is connected to the $i$-th input port of each parallel network core. Similarly, there are $N$ destination servers, where the $j$-th destination server is linked to the $j$-th output port of each core. As a result, each source server has $m$ synchronized uplinks, and each destination server has $m$ synchronized downlinks. We model the network core as a bipartite graph, with the set $\mathcal{I}$ representing the source servers on one side and the set $\mathcal{J}$ representing the destination servers on the other side. Capacity constraints apply to both the input and output ports, allowing for the transfer of one data unit per one-time unit through each port. For simplicity, we assume that the capacity of all links within each network core is uniform, meaning that all links within a core have the same speed rate.

A coflow represents a collection of independent flows that share a common performance objective. Let $\mathcal{K}$ denote the set of coflows. Each coflow $k \in \mathcal{K}$ can be represented as an $N \times N$ demand matrix $D^{(k)}$. It is important to note that each individual flow can be identified by a triple $(i,j,k)$, where $i \in \mathcal{I}$ represents the source node, $j \in \mathcal{J}$ represents the destination node, and $k \in \mathcal{K}$ denotes the corresponding coflow. The size of the flow $(i,j,k)$ is denoted as $d_{i,j,k}$, which corresponds to the $(i,j)$-th element of the demand matrix $D^{(k)}$. Each value $d_{i,j,k} \in D^{(k)}$ represents the amount of data transferred by the flow $(i,j,k)$ from input $i$ to output $j$. Furthermore, in our problem formulation, we assume that the sizes of flows are discrete and represented as integers. To simplify the problem, we consider all flows within a coflow to arrive simultaneously in the system, as described in \cite{Qiu}.

Let $C_k$ denote the completion time of coflow $k \in \mathcal{K}$. The completion time of a coflow is defined as the time when the last flow in the coflow finishes. Our objective is to schedule coflows in identical parallel networks to minimize the makespan $T = \max \limits_{\forall k \in \mathcal{K}} C_k$, which represents the maximum completion time among all coflows. Table \ref{tableN&T} provides an overview of the notation and terminology utilized in this paper.

\begin{table}[!h]
\caption{Notation and Terminology.}
\centering
\renewcommand\arraystretch{1.5}
\begin{tabular}{||c|p{5in}||}
\hline
\textbf{Symbol} & \textbf{Meaning} \\ \hline
$m$                & The number of network cores. \\ \hline
$N$                & The number of input/output ports. \\ \hline
$K$                & The number of coflows. \\ \hline
$\mathcal{M}$      & The set of network cores. $\mathcal{M}=\{1, 2, \ldots, m\}$ \\ \hline
$\mathcal{I}$      & The source sever set. $\mathcal{I}=\{1, 2, \ldots, N\}$ \\ \hline
$\mathcal{J}$      & The destination server set. $\mathcal{J}=\{1, 2, \ldots, N\}$ \\ \hline
$\mathcal{K}$      & The set of coflows. $\mathcal{K}=\{1, 2, \ldots, K\}$ \\ \hline
$\mathcal{F}$      & The set of flows from all coflows $\mathcal{K}$. \\ \hline
$D^{(k)}$          & The demand matrix of coflow $k$. \\ \hline
$d_{i,j,k}$      & The size of the flow to be transferred from input $i$ to output $j$ in coflow $k$. \\ \hline
$L_{i,k}, L_{j,k}$      & $L_{i,k} = \sum_{j=1}^{N} d_{i,j,k}$ is the total amount of data that coflow $k$ has to transfer through input port $i$, and $L_{j,k} = \sum_{i=1}^{N} d_{i,j,k}$ is the total amount of data that coflow $k$ has to transfer through output port $j$. \\ \hline
$s_h$              & The speed factor of network core $h$. \\ \hline
$C_k$              & The completion time of coflow $k$. \\ \hline
$T$                & The makespan, the maximum of the completion time of coflows. \\ \hline
\end{tabular}
\label{tableN&T}
\end{table}

\section{Approximation Algorithm for flow-level Scheduling}\label{section:approximation-algorithm-for-divisible-coflow-scheduling}
In this section, we specifically address the scenario where coflows are considered in flow level, allowing for the transfer of individual flows through different network cores. Our focus is on a solution that operates at the flow level, prohibiting flow splitting. This means that data belonging to the same flow can only be assigned to a single network core (as discussed in \cite{Weaver}).

\subsection{Algorithm}
In this subsection, we introduce two algorithms for the flow-level scheduling problem. One is flow-list-scheduling (FLS) described in Algorithm \ref{FLS}, and the other is flow-longest-processing-time-first-scheduling (FLPT) described in Algorithm \ref{FLPT}. The algorithm referred to as FLS (Algorithm \ref{FLS}) is outlined below. Let $\mathcal{F}$ represent the set of flows obtained from all coflows within the coflow set $\mathcal{K}$. For each flow $(i, j, k) \in \mathcal{F}$, our algorithm examines all flows that share congestion with $(i, j, k)$ and are scheduled prior to $(i, j, k)$. Subsequently, flow $(i, j, k)$ is assigned to the core $h \in \mathcal{M}$ with the least workload, thereby minimizing the completion time of flow $(i, j, k)$. Lines \ref{FLSfor1_s}-\ref{FLSfor1_f} determine the core with the minimum load and assign the flow to it. In terms of time complexity, Algorithm \ref{FLS} (FLS) involves scanning each flow in $\mathcal{F}$ (line \ref{FLSfor1_s} in Algorithm \ref{FLS}), which amounts to $|\mathcal{F}|$ iterations. For each flow, a comparison is made among $m$ cores to identify the least loaded core (line \ref{FLS_min} in Algorithm \ref{FLS}). Consequently, the time complexity of FLS is $O(m|\mathcal{F}|)$.

\begin{algorithm}[h]
 \caption{flow-list-scheduling} \label{FLS}
 \begin{algorithmic}[1]
   \Require a set $\mathcal{F}$, which contains of all flows $(i,j,k)$, $\forall i \in \mathcal{I}, \forall j \in \mathcal{J}, \forall k \in \mathcal{K}$ 
	 \State let $load_{I}{(i,h)}$ be the load on the $i$-th input port of the core $h$ \label{FLS_init}
   \State let $load_{O}{(j,h)}$ be the load on the $j$-th output port of the core $h$
   \State let $\mathcal{A}_h$ be the set of flows allocated to the core $h$
   \State initialize both $load_{I}$ and $load_{O}$ to 0 and $\mathcal{A}_h = \emptyset$ for all $h \in \mathcal{M}$
   \For{each flow $(i, j, k) \in \mathcal{F}$} \label{FLSfor1_s}
      \State $h^* = \arg\min_{h \in \mathcal{M}} \left\{load_{I}{(i,h)}+load_{O}{(j,h)}\right\}$ \label{FLS_min}
      \State $\mathcal{A}_{h^*} = \mathcal{A}_{h^*} \cup \{(i, j, k)\}$
      \State $load_{I}{(i,h^*)}=load_{I}{(i,h^*)} + d_{i, j, k}$ 
			\State $load_{O}{(j,h^*)}=load_{O}{(j,h^*)} + d_{i, j, k}$
   \EndFor \label{FLSfor1_f}
	 \State \textbf{return} $\left\{\mathcal{A}_h\right\}$ for $h \in \mathcal{M}$
 \end{algorithmic}
\end{algorithm}

The algorithm known as FLPT (Algorithm \ref{FLPT}) is presented below. It should be noted that Algorithm \ref{FLPT} is nearly identical to Algorithm \ref{FLS}. The key difference lies in line \ref{FLPTfor1_s} of the algorithm. Specifically, lines \ref{FLPTfor1_s}-\ref{FLPTfor1_f} involve sorting the flows $(i, j, k) \in \mathcal{F}$ in non-increasing order based on the value of $d_{i, j, k}$. Subsequently, the algorithm identifies the core with the least workload and assigns the flow to it.
In terms of time complexity, Algorithm \ref{FLPT} (FLPT) begins by spending a runtime complexity of $O(|\mathcal{F}| \log |\mathcal{F}|)$ to sort the flows (line \ref{FLPTfor1_s} in Algorithm \ref{FLPT}). Afterwards, FLPT follows the same procedure as FLS. Consequently, the time complexity of FLPT can be expressed as $O(m|\mathcal{F}| + |\mathcal{F}| \log |\mathcal{F}|)$.

\begin{algorithm}[h]
 \caption{flow-longest-processing-time-first-scheduling} \label{FLPT}
 \begin{algorithmic}[1]
   \Require a set $\mathcal{F}$, which contains of all flows $(i,j,k)$, $\forall i \in \mathcal{I}, \forall j \in \mathcal{J}, \forall k \in \mathcal{K}$
   \State let $load_{I}{(i,h)}$ be the load on the $i$-th input port of the core $h$ \label{FLPT_init}
   \State let $load_{O}{(j,h)}$ be the load on the $j$-th output port of the core $h$
   \State let $\mathcal{A}_h$ be the set of flows allocated to the core $h$
   \State initialize both $load_{I}$ and $load_{O}$ to 0 and $\mathcal{A}_h = \emptyset$ for all $h \in \mathcal{M}$
   \For{each flow $(i, j, k) \in \mathcal{F}$ in non-increasing order of $d_{i, j, k}$, breaking ties arbitrarily} \label{FLPTfor1_s}
      \State $h^* = \arg\min_{h \in \mathcal{M}} \left\{load_{I}{(i,h)}+load_{O}{(j,h)}\right\}$
      \State $\mathcal{A}_{h^*} = \mathcal{A}_{h^*} \cup \{(i, j, k)\}$
      \State $load_{I}{(i,h^*)}=load_{I}{(i,h^*)} + d_{i, j, k}$
			\State $load_{O}{(j,h^*)}=load_{O}{(j,h^*)} + d_{i, j, k}$
   \EndFor \label{FLPTfor1_f}
   \State \textbf{return} $\left\{\mathcal{A}_h\right\}$ for $h \in \mathcal{M}$
 \end{algorithmic}
\end{algorithm}

\subsection{Analysis}
This subsection shows that Algorithm \ref{FLS} achieves $(3-\tfrac{2}{m})$-approximation ratio and Algorithm \ref{FLPT} achieves $(\tfrac{8}{3}-\tfrac{2}{3m})$-approximation ratio, where $m$ is the number of network cores. An intuitive lower bound on the optimal solution cost is
\begin{eqnarray}\label{lower}
\frac{\max \left\{\max_{i} \sum_{k} L_{i, k}, \max_{j} \sum_{k} L_{j, k}\right\}}{m}
\end{eqnarray}
where $L_{i,k} = \sum_{j=1}^{N} d_{i,j,k}$, and $L_{j,k} = \sum_{i=1}^{N} d_{i,j,k}$.
First, the following lemma for Algorithm \ref{FLS} is obtained:

\begin{lemma}\label{lem_FLS}
Let $T^*$ be the cost of an optimal solution, and let $T$ denote the makespan in the schedule found by FLS (Algorithm \ref{FLS}). Then, 
\begin{eqnarray*}
T \leq \left(3-\frac{2}{m}\right) T^*.
\end{eqnarray*}
\end{lemma}
\begin{proof}
Let $F_i$ be the flow set of input port $i$, $F_j$ be the flow set of output port $j$. According to the lower bound (\ref{lower}), we know that 
\begin{align}
\frac{1}{m} \sum_{f \in F_i} d_f &\leq T^*, & &\forall i \in \mathcal{I} \label{FLS_p1} \\
\frac{1}{m} \sum_{f \in F_j} d_f &\leq T^*, & &\forall j \in \mathcal{J} \label{FLS_p2} \\
d_f &\leq T^*, & &\forall f \in \mathcal{F}. \label{FLS_p3}
\end{align}
Assume that the last flow in the schedule of FLS is the flow $f$, and the flow $f$ is sent via link $(i,j)$. We have 
\begin{eqnarray}
T  &\leq & \frac{1}{m} \sum_{f^{'} \in F_i \setminus \{f\}} d_{f^{'}} + \frac{1}{m} \sum_{f^{'} \in F_j \setminus \{f\}} d_{f^{'}} + d_f \label{FLS_p4} \\
   &\leq & 2\left(T^* - \frac{1}{m} d_f\right) + d_f \label{FLS_p5} \\
   &=    & 2T^* + \left(1-\frac{2}{m}\right)d_f \nonumber \\
   &\leq & \left(3-\frac{2}{m}\right)T^*. \label{FLS_p6}
\end{eqnarray}
The concept of inequality \eqref{FLS_p4} is similar to the proof of list scheduling in \cite{williamson_shmoys_2011}. The inequality \eqref{FLS_p5} is due to inequalities \eqref{FLS_p1} and \eqref{FLS_p2}. The inequality \eqref{FLS_p6} is based on the inequality \eqref{FLS_p3}, where $d_f \leq T^*$.
\end{proof}

Therefore, theorem \ref{thm_FLS} is derived from lemma \ref{lem_FLS}.
\begin{theorem}\label{thm_FLS}
FLS (Algorithm \ref{FLS}) has an approximation ratio of $3-\tfrac{2}{m}$, where $m$ is the number of network cores.
\end{theorem}

Next, this paper shows that Algorithm \ref{FLPT} has a better approximation ratio than Algorithm \ref{FLS}. We consider the worst case that one flow will affect other flows at input port $i \in \mathcal{I}$ and output port $j \in \mathcal{J}$ on the same core, then other affected flows will keep affecting others. This causes all flows on the same core can not be sent from input port to output port in parallel. In other words, the load of combining input port $i \in \mathcal{I}$ and output port $j \in \mathcal{J}$ of each core $h \in \mathcal{M}$ is the sum of all flows on core $h$. Therefore, the following lemma for Algorithm \ref{FLPT} is obtained:

\begin{lemma}\label{lem_FLPT}
Let $T^*$ be the cost of an optimal solution, and let $T$ denote the makespan in the schedule found by FLPT (Algorithm \ref{FLPT}). Then, 
\begin{eqnarray*}
T \leq \left(\frac{8}{3}-\frac{2}{3m}\right)T^*. 
\end{eqnarray*}
\end{lemma}
\begin{proof}
Assume that the last flow in the schedule of FLPT is the flow $f$. Considering the flows $\{1, 2, \ldots, f\} \subseteq \mathcal{F}$, they are sorted in non-increasing order of the size of flow, i.e., $d_1 \geq d_2 \geq \cdots \geq d_f$. Assume that all flows $\{1, 2, \ldots, f, f+1, \ldots, n\} = \mathcal{F}$, they are sorted in non-increasing order of the size of flow, too, i.e., $d_1 \geq d_2 \geq \cdots \geq d_f \geq d_{f+1} \geq \cdots \geq d_n$. Since flows $\{f+1,\ldots,n\} \subseteq \mathcal{F}$ do not change the value of $T$, we can omit them. Therefore, flow $f$ is viewed as the latest and the smallest flow.

Based on the discussion above, our notations can be defined. Let $\mathcal{S}\subseteq \mathcal{F}$ be the set of flows $\{1, 2, \ldots, f\}$, where $f$ is the latest and the smallest flow in the schedule of FLPT. Considering the worst case, let $T^*_{max}$ be the optimal time of the scheduling solution for all flows on input port $i \in \mathcal{I}$ and output port $j \in \mathcal{J}$, where no two flows can be transmitted simultaneously on the same core.
According Graham’s bound~\cite{Graham69}, we have 
\begin{eqnarray*}
T & \leq & \left(\frac{4}{3} - \frac{1}{3m}\right) T^*_{max}. 
\end{eqnarray*}

Let $T^*_{i}$ be the cost of an optimal solution only for port $i \in \mathcal{I}$, and $T^*_{j}$ be the optimal solution only for port $j \in \mathcal{J}$. Note that $T^*_{max} \leq T^*_{i} + T^*_{j}$ must be held, otherwise, we can construct a solution by using $T^*_{i}$ and $T^*_{j}$, which is better than the optimal solution $T^*_{max}$. Since $T^* \geq \max(T^*_{i}, T^*_{j})$, $T^*_{max} \leq T^*_{i} + T^*_{j} \leq 2T^*$. Finally, we have
\begin{eqnarray*}
T & \leq & \left(\frac{4}{3} - \frac{1}{3m}\right) T^*_{max}  \\
  & \leq & 2\left(\frac{4}{3} - \frac{1}{3m}\right) T^*  \\
  & =    & \left(\frac{8}{3} - \frac{2}{3m}\right) T^*. 
\end{eqnarray*}
\end{proof}

Therefore, theorem \ref{thm_FLPT} is derived from lemma \ref{lem_FLPT}.
\begin{theorem}\label{thm_FLPT}
FLPT (Algorithm \ref{FLPT}) has an approximation ratio of $\frac{8}{3} - \frac{2}{3m}$, where $m$ is the number of network cores.
\end{theorem}

\section{Approximation Algorithm for Coflow-level Scheduling}\label{section:approximation-algorithm-for-indivisible-coflow-scheduling}
This section considers the coflow-level scheduling problem, where distinct flows in a coflow are allowed to be transferred through the same core only. Let $L_{i,k} = \sum_{j=1}^{N} d_{i,j,k}$ be the total amount of data that coflow $k$ has to transfer through input port $i$, and $L_{j,k} = \sum_{i=1}^{N} d_{i,j,k}$ be the total amount of data that coflow $k$ has to transfer through output port $j$.

\subsection{Algorithm}
This subsection introduces an algorithm for solving the coflow-level scheduling problem. The algorithm, called coflow-list-scheduling (CLS) and described in Algorithm \ref{CLS}, aims to assign each coflow $k \in \mathcal{K}$ to a core $h \in \mathcal{M}$ in order to minimize the completion time of coflow $k$. Lines \ref{CLSfor1_s}-\ref{CLSfor1_f} of the algorithm identify the core with the minimum maximum completion time and assign the coflow to it.

In terms of time complexity, Algorithm \ref{CLS} scans each coflow in $\mathcal{K}$ (line \ref{CLSfor1_s}), resulting in $|\mathcal{K}|$ iterations. For each coflow, the algorithm compares $m$ cores to find the least loaded core (line \ref{CLS_min}). Additionally, for each core, $N^2$ pairs of input and output ports are compared to determine the maximum completion time (line \ref{CLS_min}). Consequently, the time complexity of CLS is $O(m{N^2}|\mathcal{K}|)$.

\begin{algorithm}[h]
 \caption{coflow-list-scheduling} \label{CLS}
 \begin{algorithmic}[1]
   \Require a set $\mathcal{K}$, which contains of all coflows 
   \State let $load_{I}{(i,h)}$ be the load on the $i$-th input port of the core $h$ \label{CLS_init}
   \State let $load_{O}{(j,h)}$ be the load on the $j$-th output port of the core $h$
   \State let $\mathcal{A}_h$ be the set of coflows allocated to the core $h$
   \State initialize both $load_{I}$ and $load_{O}$ to 0 and $\mathcal{A}_h = \emptyset$ for all $h \in \mathcal{M}$
   \For{each coflow $k \in \mathcal{K}$} \label{CLSfor1_s} 
      \State $h^* = \arg\min_{h \in \mathcal{M}} \max_{\forall i \in \mathcal{I}, \forall j \in \mathcal{J}} \left\{load_{I}{(i,h)}+\right.$  $\left.load_{O}{(j,h)}+L_{i,k}+L_{j,k}\right\}$ \label{CLS_min} 
      \State $\mathcal{A}_{h^*} = \mathcal{A}_{h^*} \cup \{k\}$
      \State $load_{I}{(i,h^*)}=load_{I}{(i,h^*)} + L_{i,k}$, $\forall i \in \mathcal{I}$
			\State $load_{O}{(j,h^*)}=load_{O}{(j,h^*)} + L_{j,k}$, $\forall j \in \mathcal{J}$
   \EndFor \label{CLSfor1_f}
   \State \textbf{return} $\left\{\mathcal{A}_h\right\}$ for $h \in \mathcal{M}$
 \end{algorithmic}
\end{algorithm}

\subsection{Analysis}
This section paper shows that Algorithm \ref{CLS} achieves $(2m)$-approximation ratio, where $m$ is the number of network cores. First, the following lemma for Algorithm \ref{CLS} is obtained:

\begin{lemma}\label{lem_CLS}
Let $T^*$ be the cost of an optimal solution, and let $T$ denote the makespan in the schedule found by CLS (Algorithm \ref{CLS}). Then, 
\begin{eqnarray*}
T \leq 2mT^*. 
\end{eqnarray*}
\end{lemma}
\begin{proof}
According to the lower bound (\ref{lower}), we know that 
\begin{align}
\frac{1}{m} \sum_{k \in \mathcal{K}} L_{i,k} &\leq T^*, & &\forall i \in \mathcal{I} \label{CLS_p1} \\
\frac{1}{m} \sum_{k \in \mathcal{K}} L_{j,k} &\leq T^*, & &\forall j \in \mathcal{J}. \label{CLS_p2}
\end{align}
Assume that the last flow in the schedule of CLS is sent via link $(i,j)$. We have 
\begin{align}
T &\leq \sum_{k \in \mathcal{K}} (L_{i,k} + L_{j,k}) \label{CLS_p3} \\
              &\leq mT^* + mT^* \label{CLS_p4} \\
              &= 2mT^*. \nonumber
\end{align}
The inequality \eqref{CLS_p3} is held since $T$ is bounded by the size of all flows via link $(i,j)$. The inequality \eqref{CLS_p4} is due to inequalities \eqref{CLS_p1} and \eqref{CLS_p2}, where $\sum \limits_{k \in \mathcal{K}} L_{i,k} \leq mT^*$ and $\sum \limits_{k \in \mathcal{K}} L_{j,k} \leq mT^*$.
\end{proof}

Therefore, theorem \ref{thm_CLS} is derived from lemma \ref{lem_CLS}.
\begin{theorem}
\label{thm_CLS}CLS (Algorithm \ref{CLS}) has an approximation ratio of $(2m)$, where $m$ is the number of network cores.
\end{theorem}

\section{Experiments}\label{section:experiments}
This section presents the simulation results and evaluates the performance of our proposed algorithms. Additionally, this paper compares the performance of our proposed algorithms for the flow-level scheduling problem with the algorithm Weaver proposed by Huang \textit{et al.} \cite{Weaver}. Furthermore, the experiment demonstrates that our results align with the approximation ratios analyzed in sections \ref{section:approximation-algorithm-for-divisible-coflow-scheduling} and \ref{section:approximation-algorithm-for-indivisible-coflow-scheduling}.

\subsection{Workload}
We have implemented a flow-level simulator to track the assignment of each flow to various cores in both identical parallel networks and heterogeneous parallel networks. Our simulator is based on Mosharaf's implementation \cite{website:mosharaf}, which originally simulates coflows on a single core. To track coflows assigned to $m$ cores, we have modified the code so that our simulator traces flows $m$ times for all cores. Additionally, we have incorporated Shafiee and Ghaderi's algorithm \cite{Shafiee} to ensure that all flows are transferred in a preemptible manner within each core. Moreover, each link in our simulator has a capacity of 128 MBps. We have chosen the time unit to be $\tfrac{1}{128}$ second (approximately 8 milliseconds) so that each link has a capacity of 1 MB per time unit.

In our study, all algorithms are simulated using both synthetic and real traffic traces. In synthetic traces, coflows are generated based on the number of coflows, denoted as $K$, and the number of ports, denoted as $N$. Each coflow is described by $(W_{min}, W_{max}, L_{min}, L_{max})$, where $1 \leq W_{min} \leq W_{max}$ and $1 \leq L_{min} \leq L_{max}$, for all $W_{min}, W_{max}, L_{min}, L_{max} \in \mathbb{Z}$. Let $M$ represent the number of non-zero flows within each coflow. Then, $M = w_1 \cdot w_2$, where $w_1$ and $w_2$ are randomly chosen from the set $\{W_{min}, W_{min}+1, \ldots, W_{max}\}$. Additionally, $w_1$ is randomly assigned to input links, and $w_2$ is randomly assigned to output links. The size of flow $d_{i,j,k}$ is randomly selected from $\{L_{min}, L_{min}+1, \ldots, L_{max}\}$. If the construction of a coflow is not explicitly specified in the synthetic traces, the default construction for all coflows follows a certain distribution of coflow descriptions: $(1, 5, 1, 10)$, $(1, 5, 10, 1000)$, $(5, N, 1, 10)$, and $(5, N, 10, 1000)$ with percentages of $41\%$, $29\%$, $9\%$, and $21\%$ respectively.

In real traces, coflows are generated from a realistic workload based on a Hive/MapReduce trace \cite{website:benchmark} obtained from Facebook, which was collected from a 3000-machine setup with 150 racks. These real traces have been widely used as benchmarks in various works, such as \cite{Varys, Weaver, Qiu, Shafiee}. The purpose of this benchmark is to provide realistic workloads synthesized from real-world data-intensive applications for the evaluation of coflow-based solutions. Since this paper does not consider release times, the release time for all coflows is set to 0.

In order to assess the performance of the algorithms, we calculate the approximation ratio. For identical parallel networks, the ratio is obtained by dividing the makespan achieved by the algorithms by the lower bound of the optimal value:
\begin{eqnarray*}
opt = \frac{\max \left\{\max_{i} \sum_{k} L_{i, k}, \max_{j} \sum_{k} L_{j, k}\right\}}{m}
\end{eqnarray*}
where $m$ represents the number of network cores. This lower bound value provides an estimate of the optimal makespan.

For heterogeneous parallel networks, the ratio of algorithms is calculated by dividing the makespan obtained from the algorithms by another lower bound of the optimal value:
\begin{eqnarray*}
opt = \frac{\max \left\{\max_{i} \sum_{k} L_{i, k}, \max_{j} \sum_{k} L_{j, k}\right\}}{\sum_{l=1}^{m} s_{l}}
\end{eqnarray*}
In this case, $s_{l}$ represents the speed factor of core $l$. In a heterogeneous parallel network with $m$ network cores, $s_{l}$ is randomly selected from the range $\left[1, \frac{m}{h}\right]$, where $h$ is a heterogeneity factor. Higher values of $h$ indicate lower variance in network core speeds. This lower bound estimation takes into account the varying speeds of the cores.

By comparing the achieved makespan with these lower bound estimates, we can evaluate the performance of the algorithms in both identical and heterogeneous parallel networks.

\subsection{Simulation Results in Identical Parallel Networks}

\begin{figure}[!h]
\centering
\subfigure[The performance of algorithms: FLS, FLPT, and Weaver.]{
\begin{minipage}[h]{0.4\textwidth}
\centering
{\includegraphics[width=3.4in]{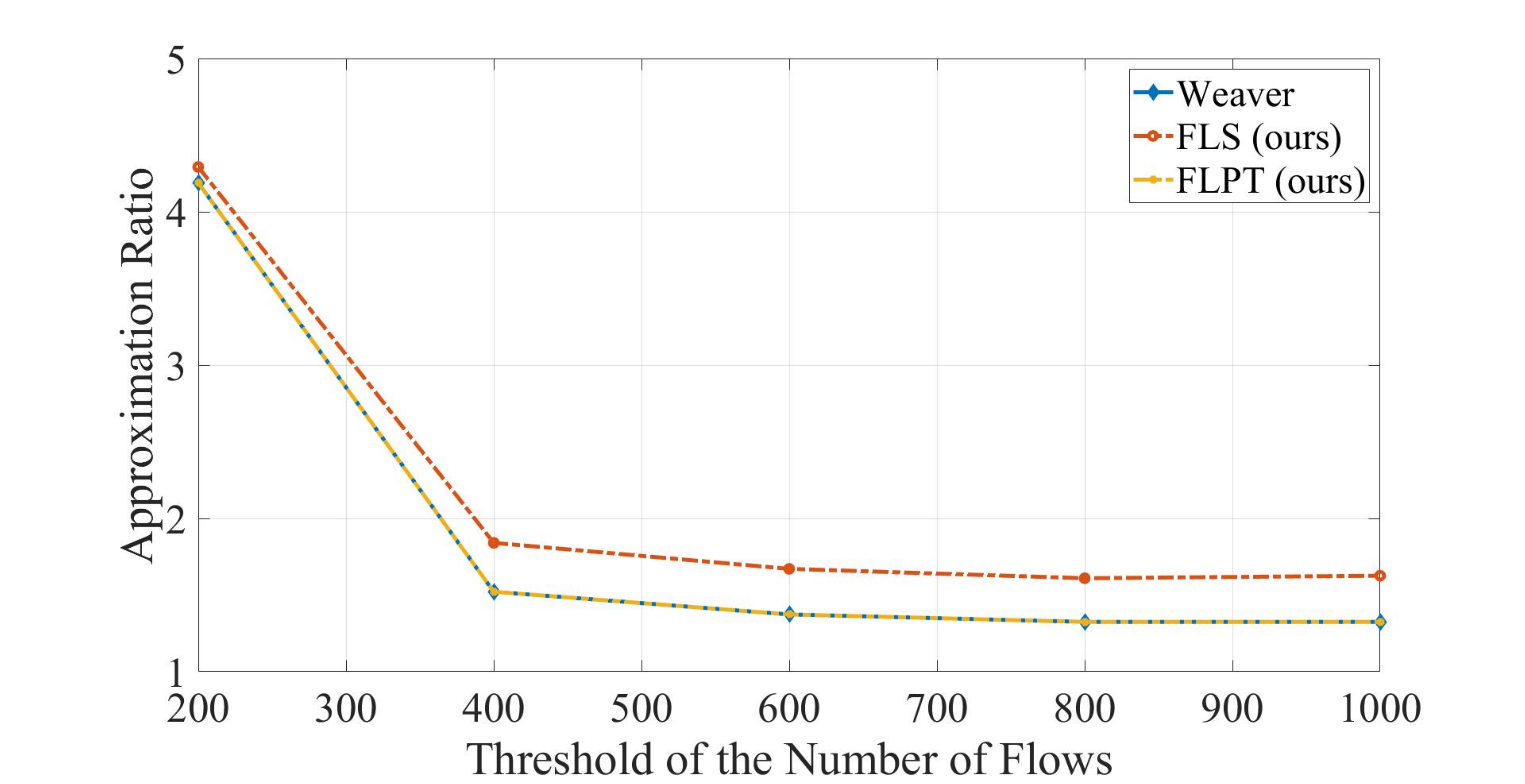}}
\end{minipage}
\label{fig:Divisible coflows from benchmark}
}

\subfigure[The performance of algorithm: CLS.]{
\begin{minipage}[h]{0.4\textwidth}
\centering
{\includegraphics[width=3.4in]{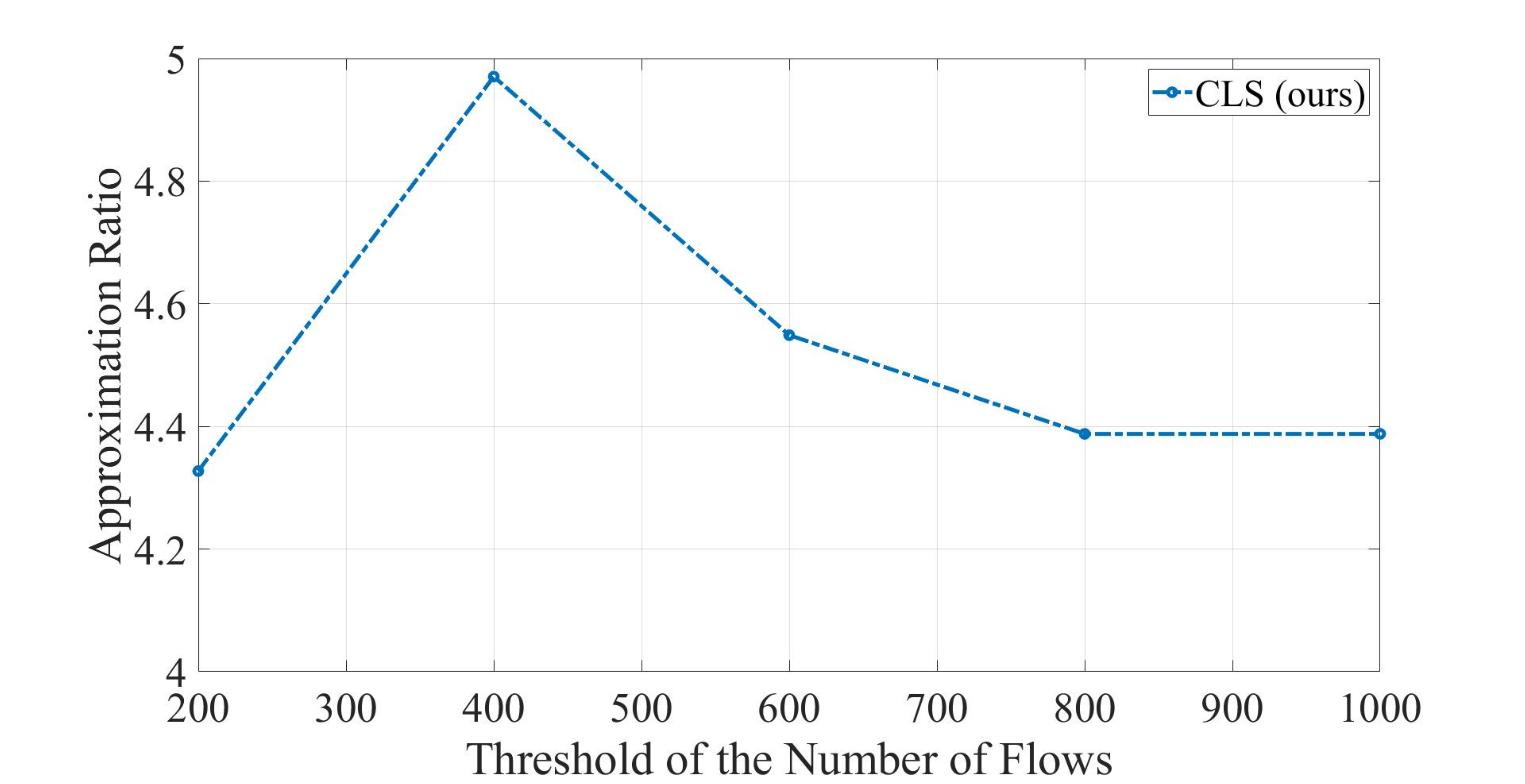}}
\end{minipage}
\label{fig:Indivisible coflows from benchmark}
}
\caption{The approximation ratios of FLS, FLPT, Weaver, and CLS for various thresholds of the number of flows using real traces in identical parallel networks.}
\label{fig:Coflows from benchmark}
\end{figure}

Figure \ref{fig:Coflows from benchmark} illustrates the algorithm ratios of FLS, FLPT, Weaver, and CLS for different thresholds of the number of flows in identical parallel networks. The real traces consist of 526 coflows distributed across $m = 5$ network cores with $N = 150$ input/output links. Among all the coflows, the maximum number of flows is 21170, while the minimum number is 1. Additionally, the maximum flow size is 2472 MB, and the minimum size is 1 MB.
Similar to the approach described in \cite{Shafiee}, we set a threshold to filter coflows based on the number of non-zero flows. Coflows with a number of flows below the threshold are filtered out. We consider five collections filtered using the following thresholds: 200, 400, 600, 800, and 1000.

Subsequent experiments reveal that when there is a large number of coflows or sparse demand matrices, FLPT and Weaver demonstrate very similar performance. Consequently, in Figure \ref{fig:Divisible coflows from benchmark}, FLPT exhibits the same ratio as Weaver. Additionally, CLS matches the ratio of $2m$ as depicted in Figure \ref{fig:Indivisible coflows from benchmark}.

\begin{figure}[!h]
\centering
\subfigure[The performance of algorithms: FLS, FLPT, and Weaver.]{
\begin{minipage}[h]{0.4\textwidth}
\centering
{\includegraphics[width=3.4in]{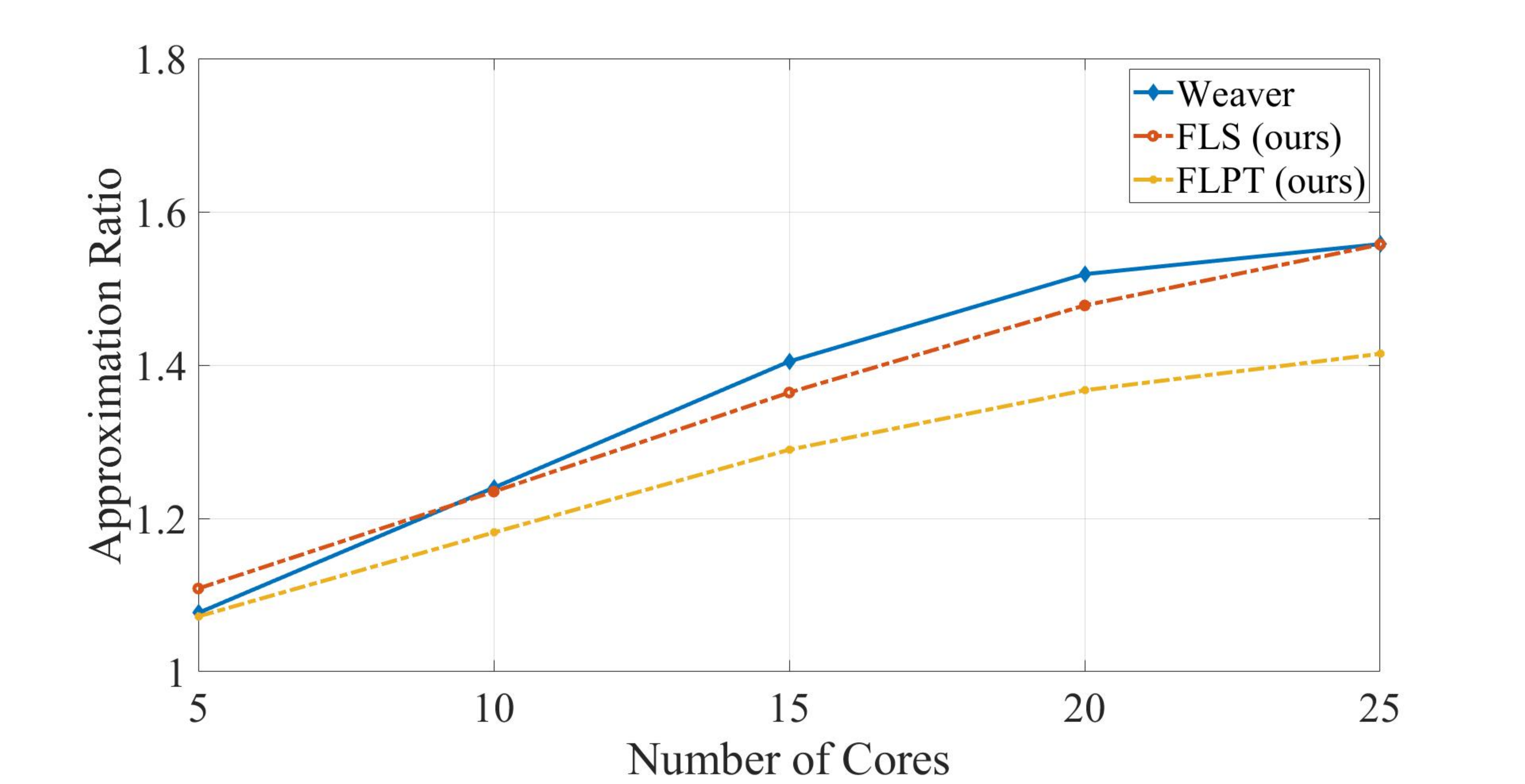}}
\end{minipage}
\label{fig:Divisible coflows from custom num of core}
}

\subfigure[The performance of algorithm: CLS.]{
\begin{minipage}[h]{0.4\textwidth}
\centering
{\includegraphics[width=3.4in]{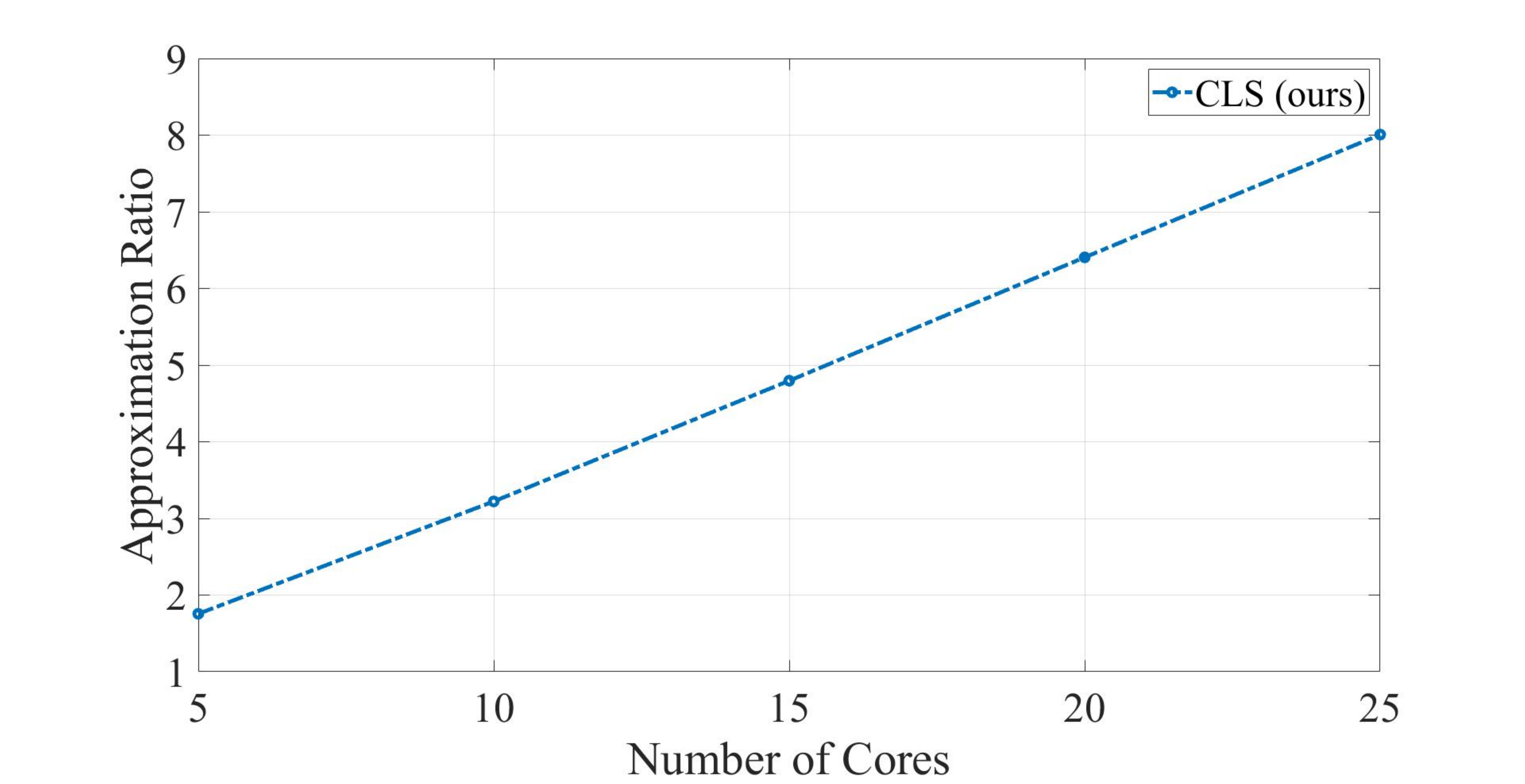}}
\end{minipage}
\label{fig:Indivisible coflows from custom num of core}
}
\caption{Approximation ratio of FLS, FLPT, Weaver, and CLS for various number of cores under synthetic traces in identical parallel networks.}
\label{fig:Coflows from custom num of core}
\end{figure}

Figure \ref{fig:Coflows from custom num of core} illustrates the algorithm ratios of FLS, FLPT, Weaver, and CLS for different numbers of network cores in identical parallel networks. In this synthetic trace, we consider $25$ coflows across 5 scenarios with varying numbers of network cores and $N=10$ input/output links. Each scenario represents a distinct number of cores, namely $m=5, 10, 15, 20, 25$. For each scenario, we generate 100 sample traces and report the average performance of the algorithms.

The results demonstrate that FLPT consistently outperforms Weaver in terms of the ratio, as depicted in Figure \ref{fig:Divisible coflows from custom num of core}. Furthermore, as the number of cores increases, the improvement provided by FLPT becomes more significant.
Additionally, the ratio of CLS exhibits an increasing trend with the growing number of cores, as illustrated in Figure \ref{fig:Indivisible coflows from custom num of core}. This result aligns with the theoretical analysis.

\begin{figure}[!h]
\centering
\subfigure[The performance of algorithms: FLS, FLPT, and Weaver.]{
\begin{minipage}[h]{0.4\textwidth}
\centering
{\includegraphics[width=3.4in]{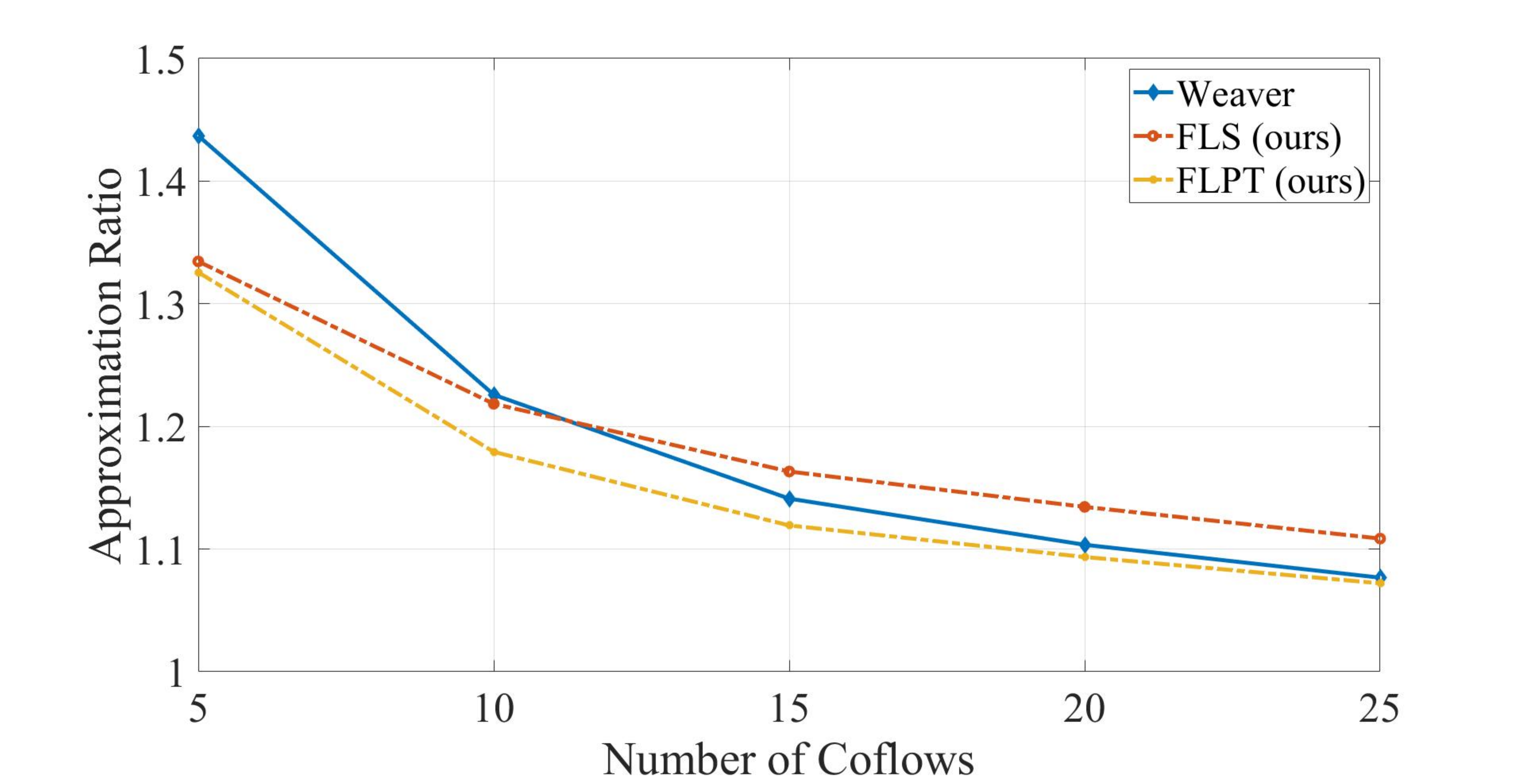}}
\end{minipage}
\label{fig:Divisible coflows from custom num of coflows}
}

\subfigure[The performance of algorithm: CLS.]{
\begin{minipage}[h]{0.4\textwidth}
\centering
{\includegraphics[width=3.4in]{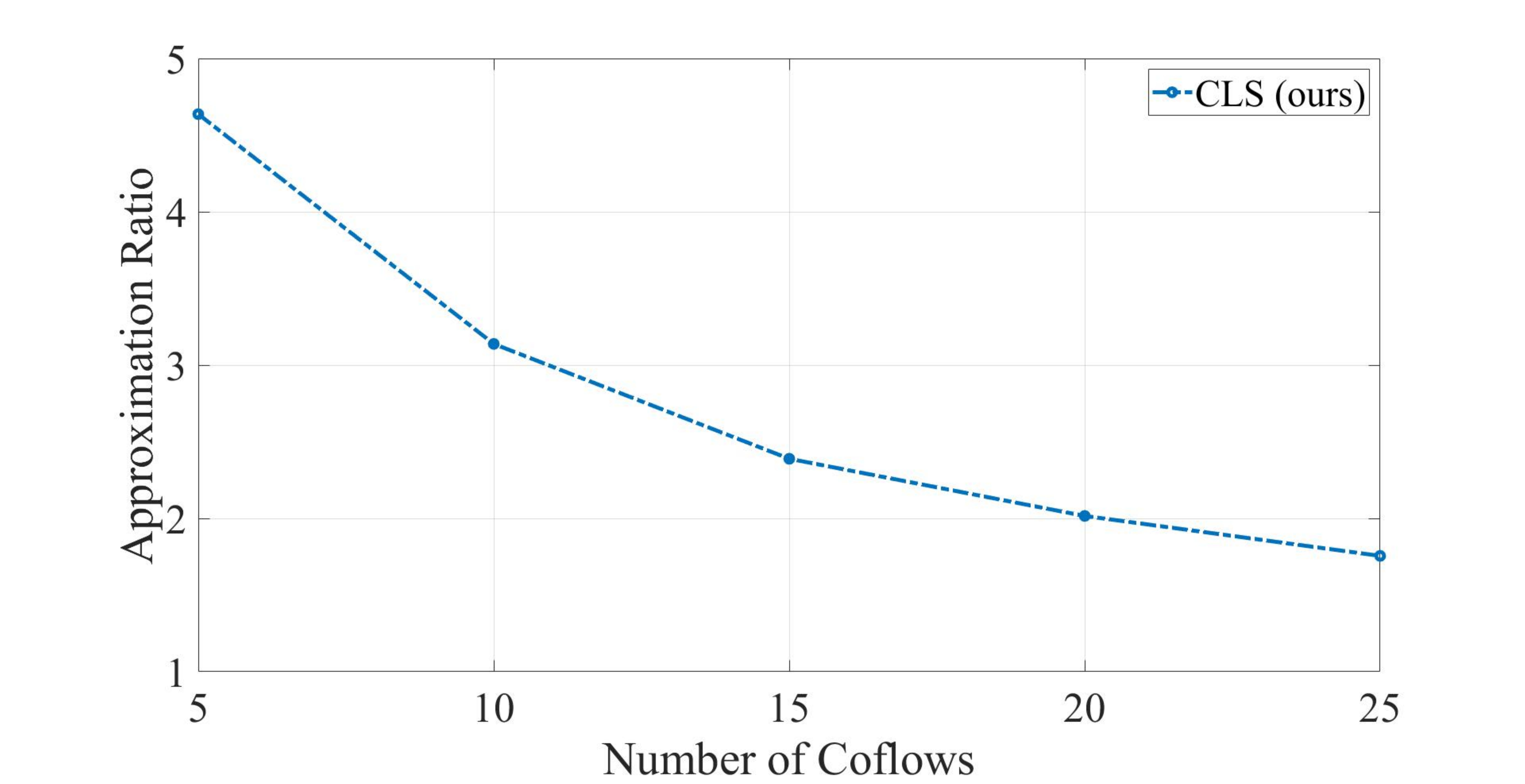}}
\end{minipage}
\label{fig:Indivisible coflows from custom num of coflows}
}
\caption{Approximation ratio of FLS, FLPT, Weaver, and CLS for various number of coflows under synthetic traces in an identical parallel network.}
\label{fig:Coflows from custom num of coflows}
\end{figure}

Figure \ref{fig:Coflows from custom num of coflows} illustrates the algorithm ratios of FLS, FLPT, Weaver, and CLS for different numbers of network coflows in an identical parallel network. In this synthetic trace, we consider $m=5$ cores with $N=10$ input/output links across 5 scenarios with varying numbers of coflows. Each scenario represents a distinct number of coflows: $5, 10, 15, 20, 25$. For each scenario, we generate 100 sample traces and report the average performance of the algorithms.

The results demonstrate that FLPT outperforms Weaver in terms of the approximation ratio, as depicted in Figure \ref{fig:Divisible coflows from custom num of coflows}. Furthermore, as the number of coflows increases, Weaver's performance approaches that of FLPT.
Moreover, the ratio of CLS decreases as the number of coflows increases, as shown in Figure \ref{fig:Indivisible coflows from custom num of coflows}. This indicates that the algorithm performs better with a larger number of coflows.

\begin{figure}[!h]
\centering
\subfigure[The performance of algorithms: FLS, FLPT, and Weaver.]{
\begin{minipage}[h]{0.4\textwidth}
\centering
{\includegraphics[width=3.4in]{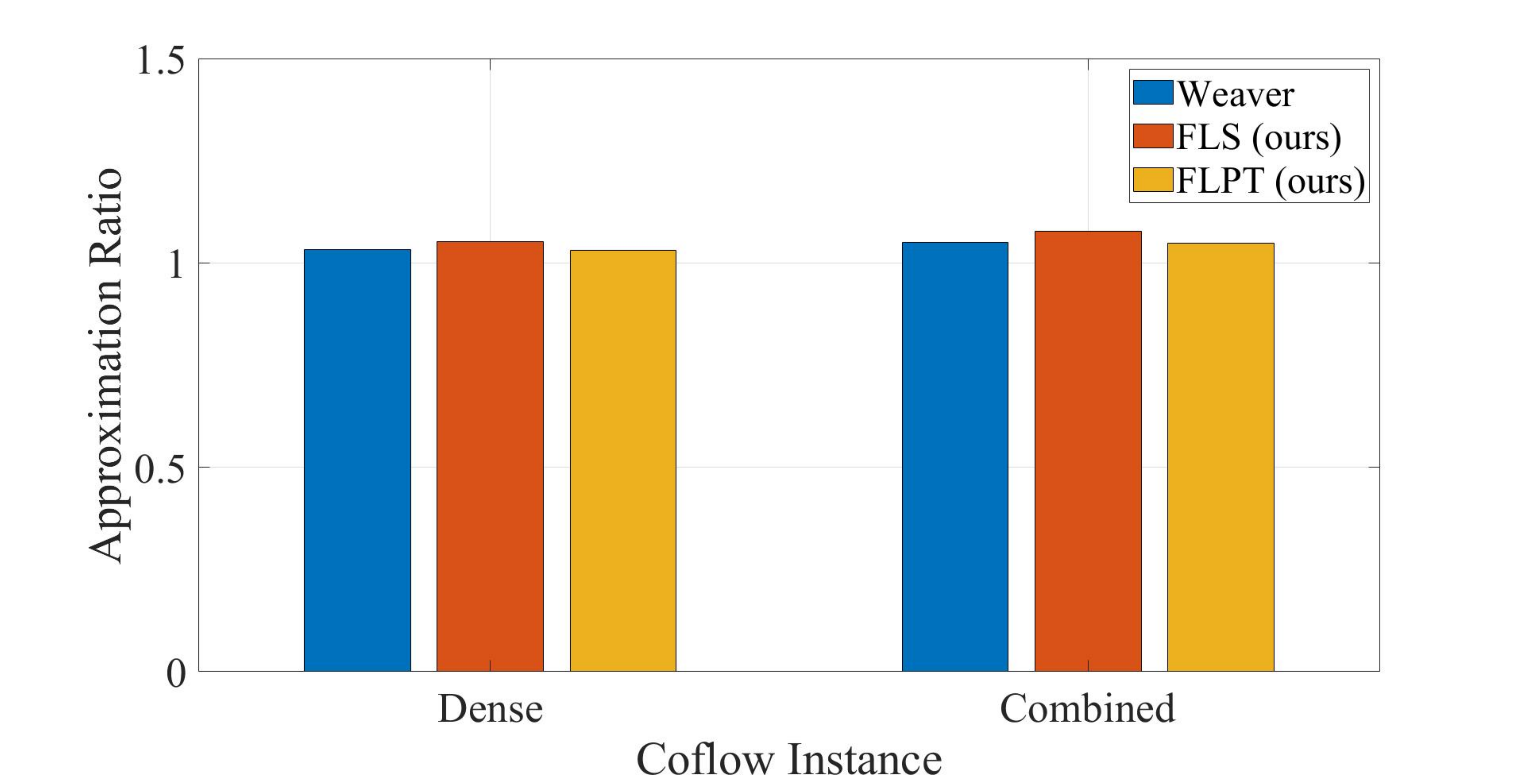}}
\end{minipage}
\label{fig:Divisible coflows from custom dense and combined}
}

\subfigure[The performance of algorithm: CLS.]{
\begin{minipage}[h]{0.4\textwidth}
\centering
{\includegraphics[width=3.4in]{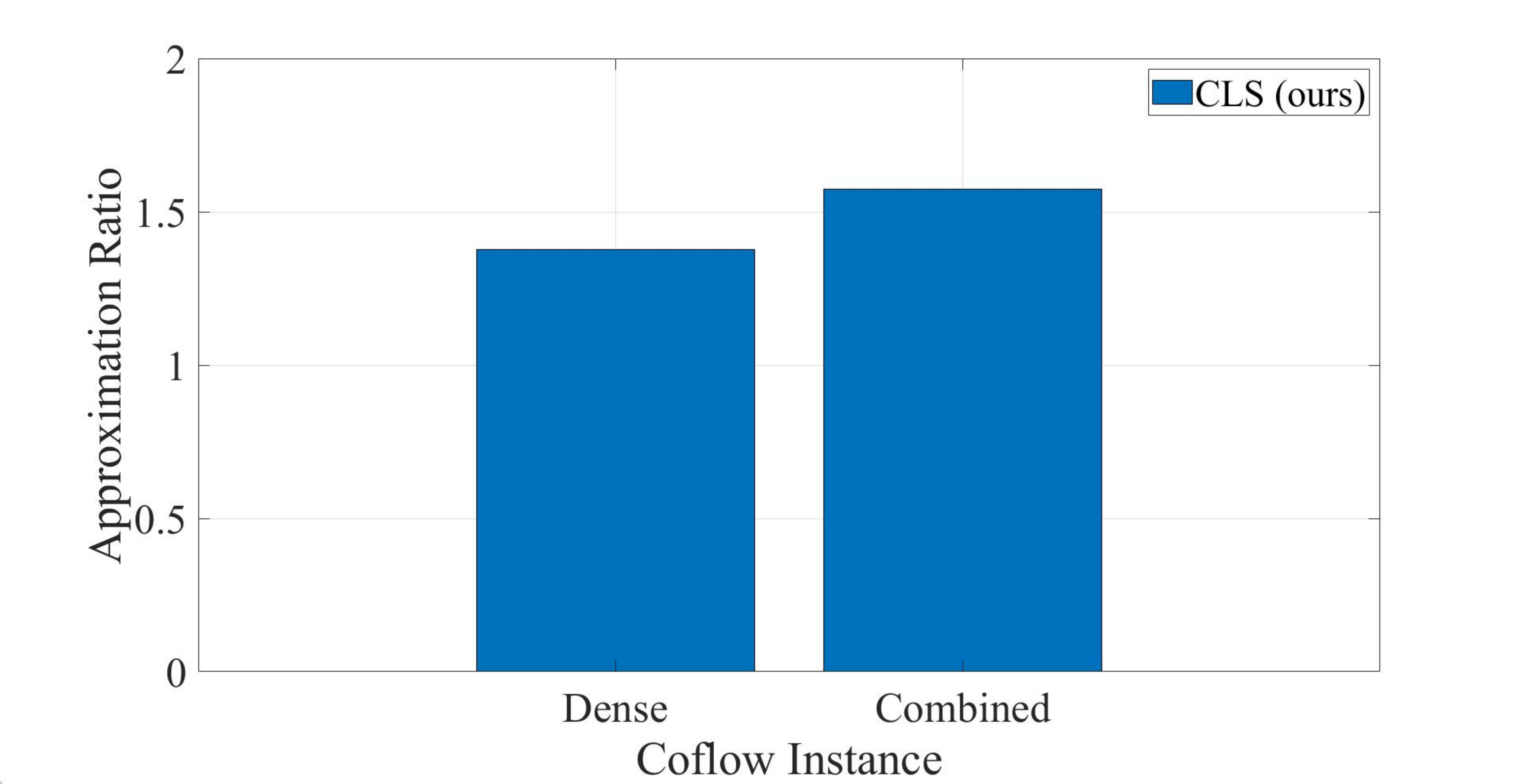}}
\end{minipage}
\label{fig:Indivisible coflows from custom dense and combined}
}
\caption{Approximation ratio of FLS, FLPT, Weaver, and CLS for dense and combined instances under synthetic traces in an identical parallel network.}
\label{fig:Coflows from custom dense and combined}
\end{figure}

Figure \ref{fig:Coflows from custom dense and combined} illustrates the approximation ratio of FLS, FLPT, Weaver, and CLS for dense and combined instances, as described in \cite{Shafiee}. These instances are deployed in an identical parallel network. The synthetic trace consists of two sets of 25 coflows, each utilizing $m=5$ network cores and $N=10$ input/output links. One set represents a dense instance, while the other represents a combined instance.
To create dense and combined instances, we define dense and sparse coflows. In a dense coflow, the coflow description $(W_{min}, W_{max}, L_{min}, L_{max})$ is set to $(\sqrt{N}, N, 1, 100)$. On the other hand, a sparse coflow has the description $(1, \sqrt{N}, 1, 100)$. Therefore, in a dense instance, each coflow is dense, whereas in a combined instance, each coflow has an equal probability of being dense or sparse.

We generate 100 sample traces for each instance and present the average performance of the algorithms. The results show that FLPT outperforms Weaver in both dense and combined instances, as depicted in Figure \ref{fig:Divisible coflows from custom dense and combined}. The improvement is more significant in the case of dense instances. Moreover, the ratio of dense instances is superior to that of combined instances, as shown in Figure \ref{fig:Divisible coflows from custom dense and combined} and Figure \ref{fig:Indivisible coflows from custom dense and combined}.

\begin{figure}[!h]
\centering
\subfigure[The box plot of algorithms: FLS, FLPT, and Weaver.]{
\begin{minipage}[h]{0.4\textwidth}
\centering
{\includegraphics[width=3.4in]{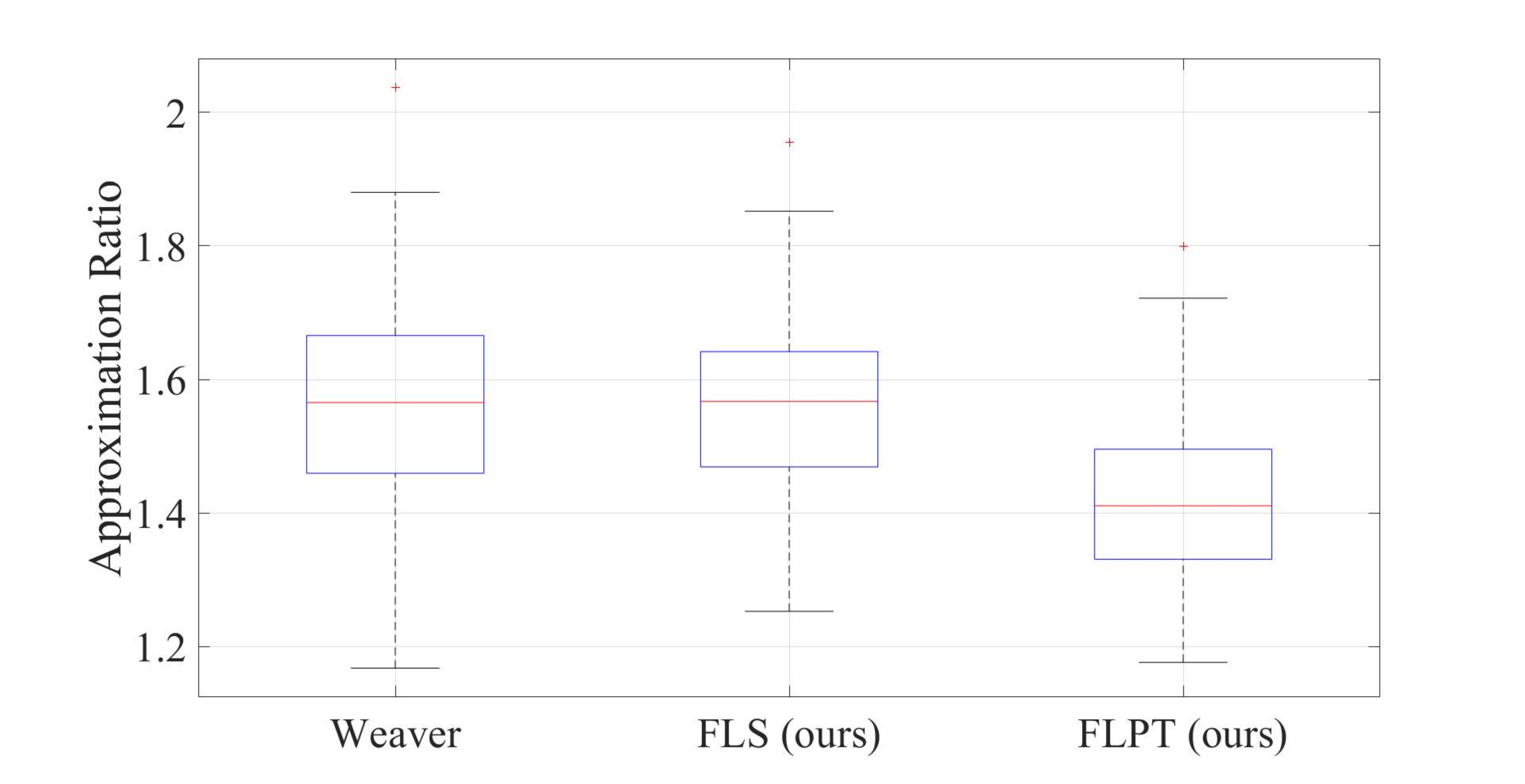}}
\end{minipage}
\label{fig:Divisible coflows from custom box plot}
}

\subfigure[The box plot of algorithm: CLS.]{
\begin{minipage}[h]{0.4\textwidth}
\centering
{\includegraphics[width=3.4in]{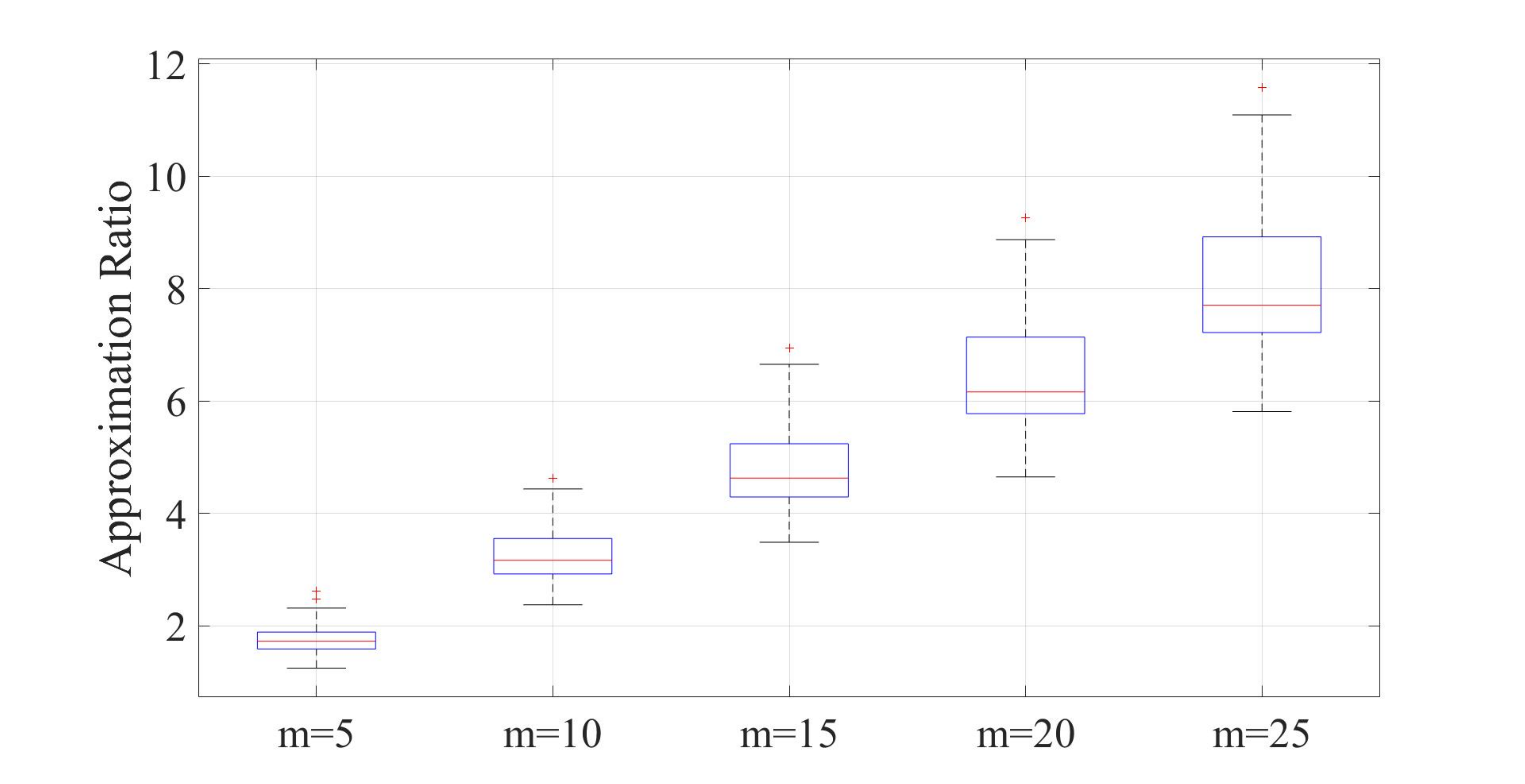}}
\end{minipage}
\label{fig:Indivisible coflows from custom box plot}
}
\caption{The box plot of FLS, FLPT, Weaver, and CLS under synthetic traces in an identical parallel network.}
\label{fig:Coflows from custom box plot}
\end{figure}

\begin{table}[h]
\caption{The quartiles, maximum, and minimum of FLS, FLPT, Weaver, and CLS for box plots in Figure \ref{fig:Coflows from custom box plot}.}
\centering
\renewcommand\arraystretch{1.5}
\begin{tabular}{||c|c|c|c|c||}
\hline
\textbf{\diagbox{Q{\&}M}{Algo}} & \textbf{FLS} & \textbf{FLPT} & \textbf{Weaver} & \textbf{CLS} \\ \hline
\textbf{$Q_1$} & 1.4692 & 1.3310 & 1.4597 & 7.2222 \\ \hline
\textbf{$Q_2$} & 1.5671 & 1.4109 & 1.5655 & 7.7068 \\ \hline
\textbf{$Q_3$} & 1.6418 & 1.4956 & 1.6659 & 8.9244 \\ \hline
\textbf{Maximum} & 1.9551 & 1.7989 & 2.0372 & 11.5818 \\ \hline
\textbf{Minimum} & 1.2529 & 1.1764 & 1.1679 & 5.8152 \\ \hline
\end{tabular}
\label{tableBoxPlot}
\end{table}

Figure \ref{fig:Coflows from custom box plot} presents a box plot showing the performance of FLS, FLPT, Weaver, and CLS in an identical parallel network. The synthetic trace used for this analysis consists of $25$ coflows deployed in a network with $m=25$ cores and $N=10$ input/output links.
We conducted 100 sample traces for each algorithm and represented the results using a box plot. The plot includes quartiles, maximum and minimum values for each algorithm. The findings demonstrate that FLPT not only achieves a superior ratio compared to Weaver but also exhibits a narrower interquartile range. This is evident in Figure \ref{fig:Coflows from custom box plot} and summarized in Table \ref{tableBoxPlot}. In Figure \ref{fig:Indivisible coflows from custom box plot}, it can be observed that the approximation ratio of CLS, as well as its interquartile range, increase as the number of cores grows.

\begin{figure}[!h]
\centering
\subfigure[The CDF of algorithms: FLS, FLPT, and Weaver.]{
\begin{minipage}[h]{0.4\textwidth}
\centering
{\includegraphics[width=3.4in]{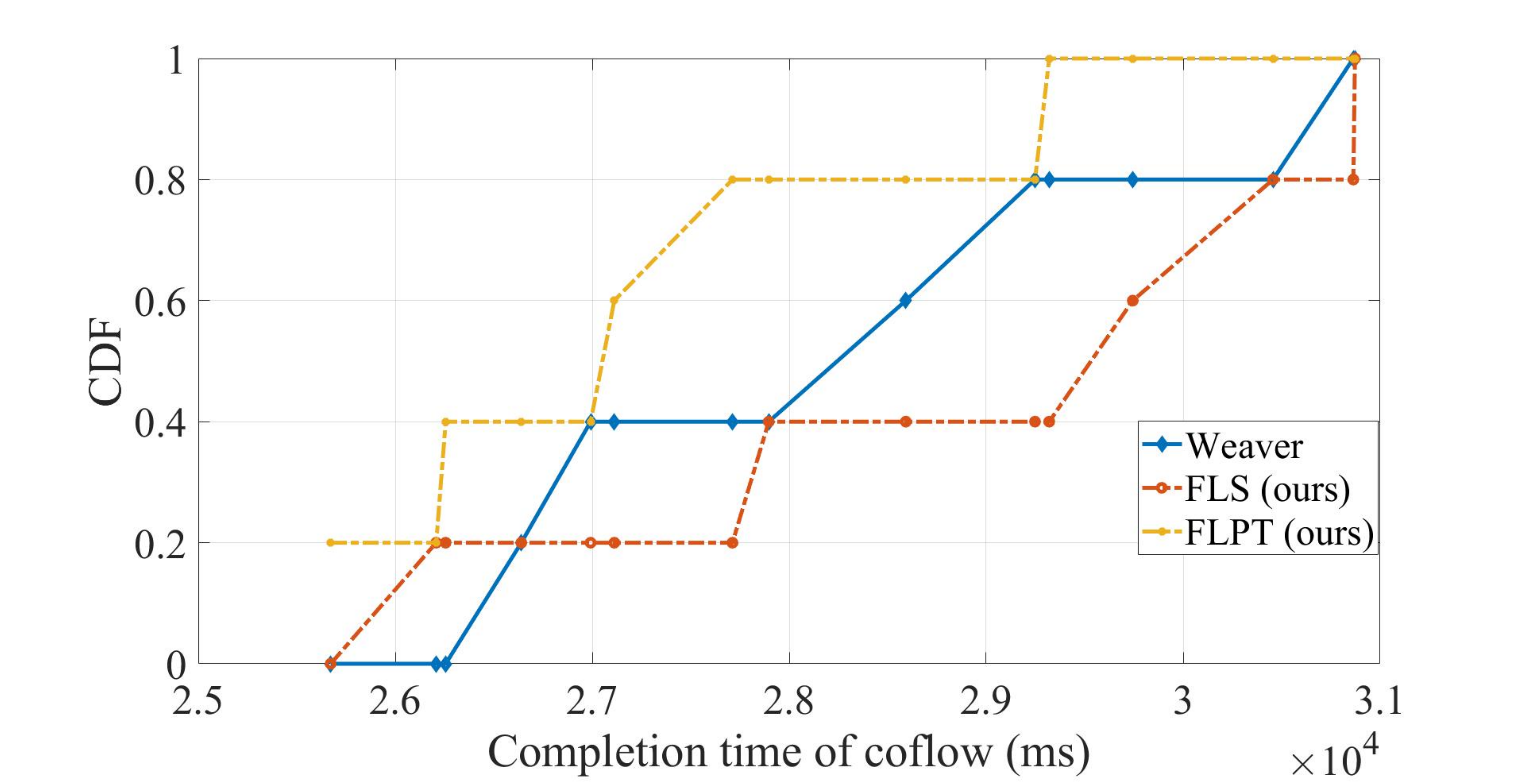}}
\end{minipage}
\label{fig:Divisible coflows from custom CDF}
}

\subfigure[The CDF of algorithm: CLS.]{
\begin{minipage}[h]{0.4\textwidth}
\centering
{\includegraphics[width=3.4in]{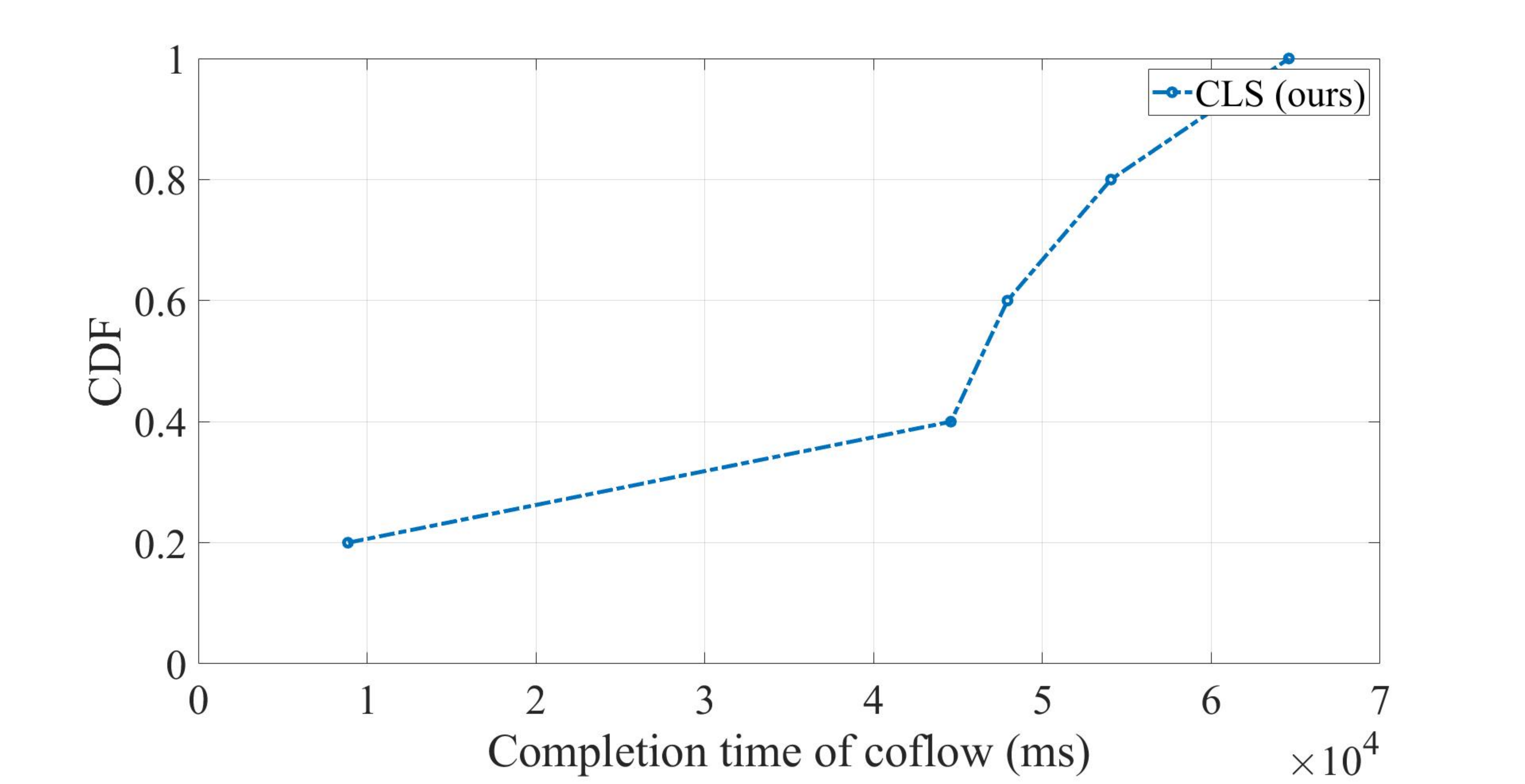}}
\end{minipage}
\label{fig:Indivisible coflows from custom CDF}
}
\caption{CDF of the core completion time for FLS, FLPT, Weaver, and CLS under synthetic traces in an identical parallel network.}
\label{fig:Coflows from custom CDF}
\end{figure}

Figure \ref{fig:Coflows from custom CDF} displays the cumulative distribution function (CDF) of core completion time for FLS, FLPT, Weaver, and CLS in an identical parallel network. The synthetic trace used in this analysis involves 15 coflows operating in a network with $m=5$ cores and $N=10$ input/output links.
The results indicate that FLPT achieves core completion before 29.320 seconds for all cores, while both FLS and Weaver achieve completion before 30.872 seconds, as illustrated in Figure \ref{fig:Divisible coflows from custom CDF}. Additionally, CLS completes the core processing before 64.592 seconds for all cores, as depicted in Figure \ref{fig:Indivisible coflows from custom CDF}.

\begin{figure}[!h]
\centering
\subfigure[The performance of algorithms: FLPT, and Weaver.]{
\begin{minipage}[h]{0.4\textwidth}
\centering
{\includegraphics[width=3.4in]{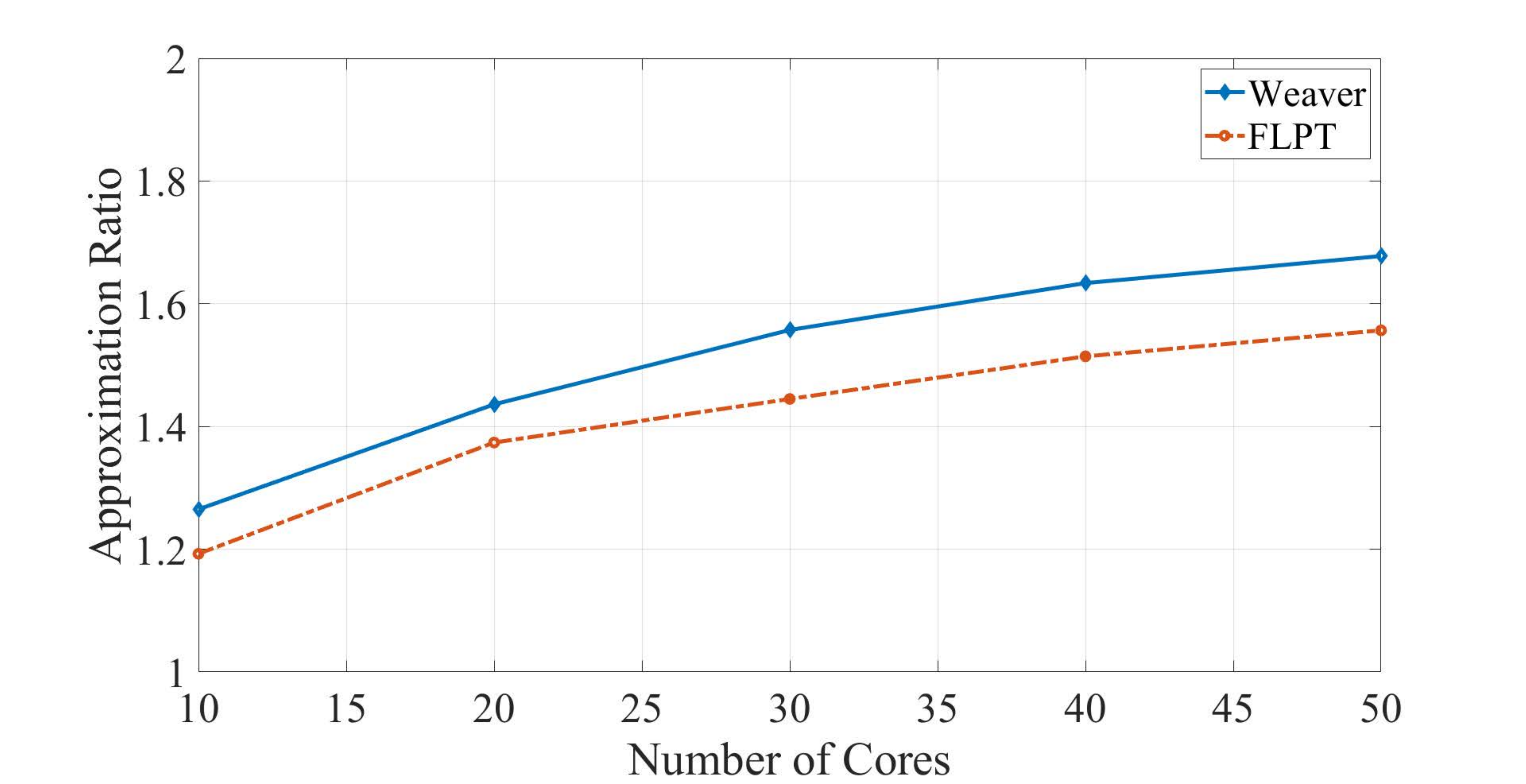}}
\end{minipage}
\label{fig:Divisible coflows from custom num of cores h}
}

\subfigure[The performance of algorithm: CLS.]{
\begin{minipage}[h]{0.4\textwidth}
\centering
{\includegraphics[width=3.4in]{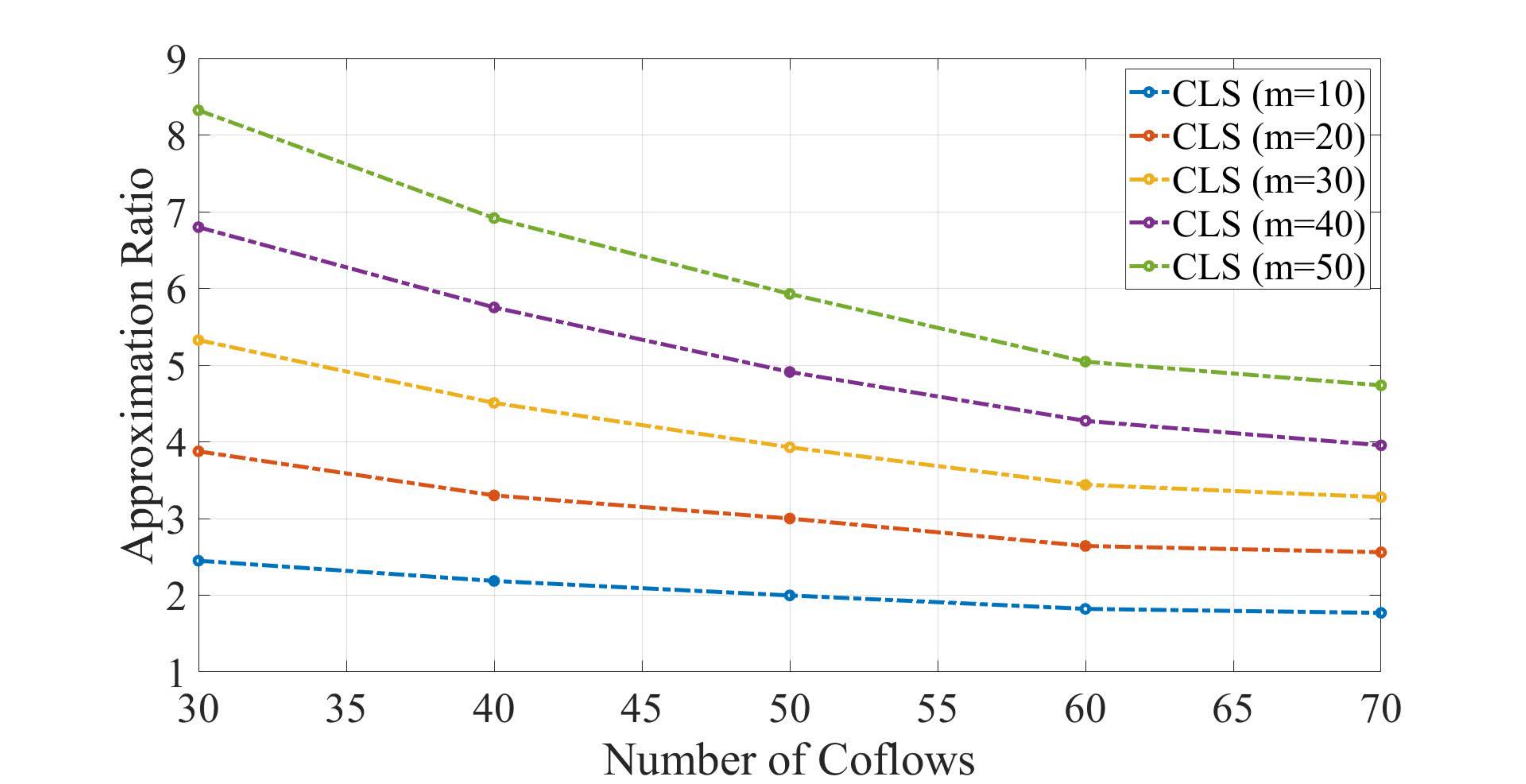}}
\end{minipage}
\label{fig:Indivisible coflows from custom num of cores h}
}
\caption{Approximation ratio of FLPT, Weaver, and CLS for distinct number of cores under synthetic traces in heterogeneous parallel networks.}
\label{fig:Coflows from custom num of cores h}
\end{figure}

\subsection{Simulation Results in Heterogeneous Parallel Networks}
The algorithm proposed in this paper can be adapted for scheduling in heterogeneous parallel networks. The pseudocode of the algorithms is provided in the APPENDIX. The following are the simulation results of the algorithms in heterogeneous parallel networks.

Figure \ref{fig:Coflows from custom num of cores h} illustrates the approximation ratio of FLPT and Weaver for different numbers of cores in heterogeneous parallel networks. The analysis is based on a synthetic trace comprising 25 coflows in five scenarios with varying numbers of network cores, while maintaining $N=10$ input/output links.
In Figure \ref{fig:Divisible coflows from custom num of cores h}, we set $h$ to 5, while in Figure \ref{fig:Indivisible coflows from custom num of cores h}, we set $h$ to 1. For each scenario, we generate 100 sample traces and report the average performance of the algorithms.

The findings demonstrate that as the number of cores increases, the approximation ratio also increases. Furthermore, as the number of cores increases, the performance gap between FLPT and Weaver widens. These observations align with the results observed in identical parallel networks.
In Figure \ref{fig:Indivisible coflows from custom num of cores h}, the approximation ratio of CLS increases with the number of cores and decreases with the number of coflows. This result is also consistent with the findings observed in identical parallel networks.

\begin{figure}[!h]
\centering
\subfigure[The box plot of algorithms: FLS, FLPT, and Weaver.]{
\begin{minipage}[h]{0.4\textwidth}
\centering
{\includegraphics[width=3.4in]{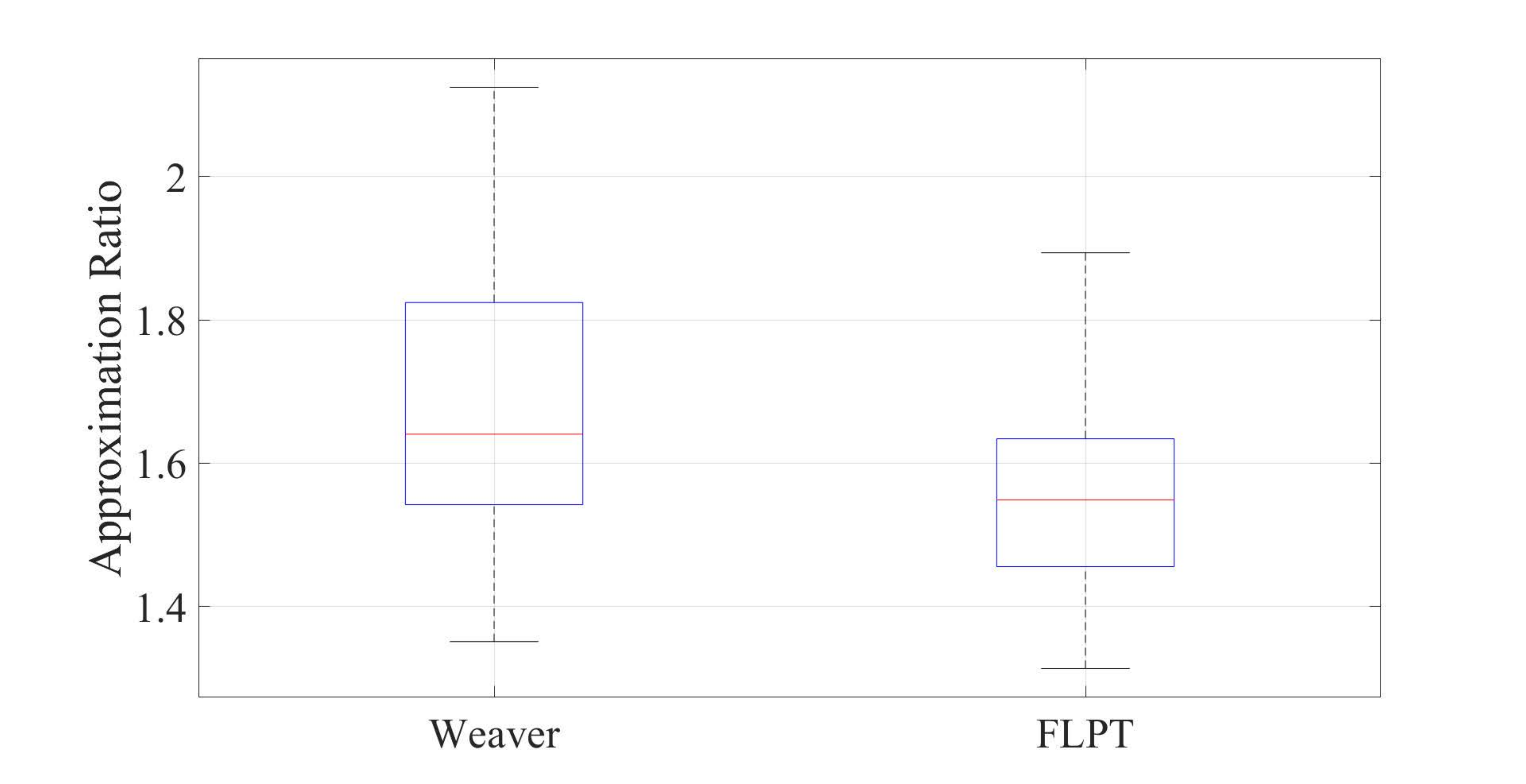}}
\end{minipage}
\label{fig:Divisible coflows from custom box plot h}
}

\subfigure[The box plot of algorithm: CLS.]{
\begin{minipage}[h]{0.4\textwidth}
\centering
{\includegraphics[width=3.4in]{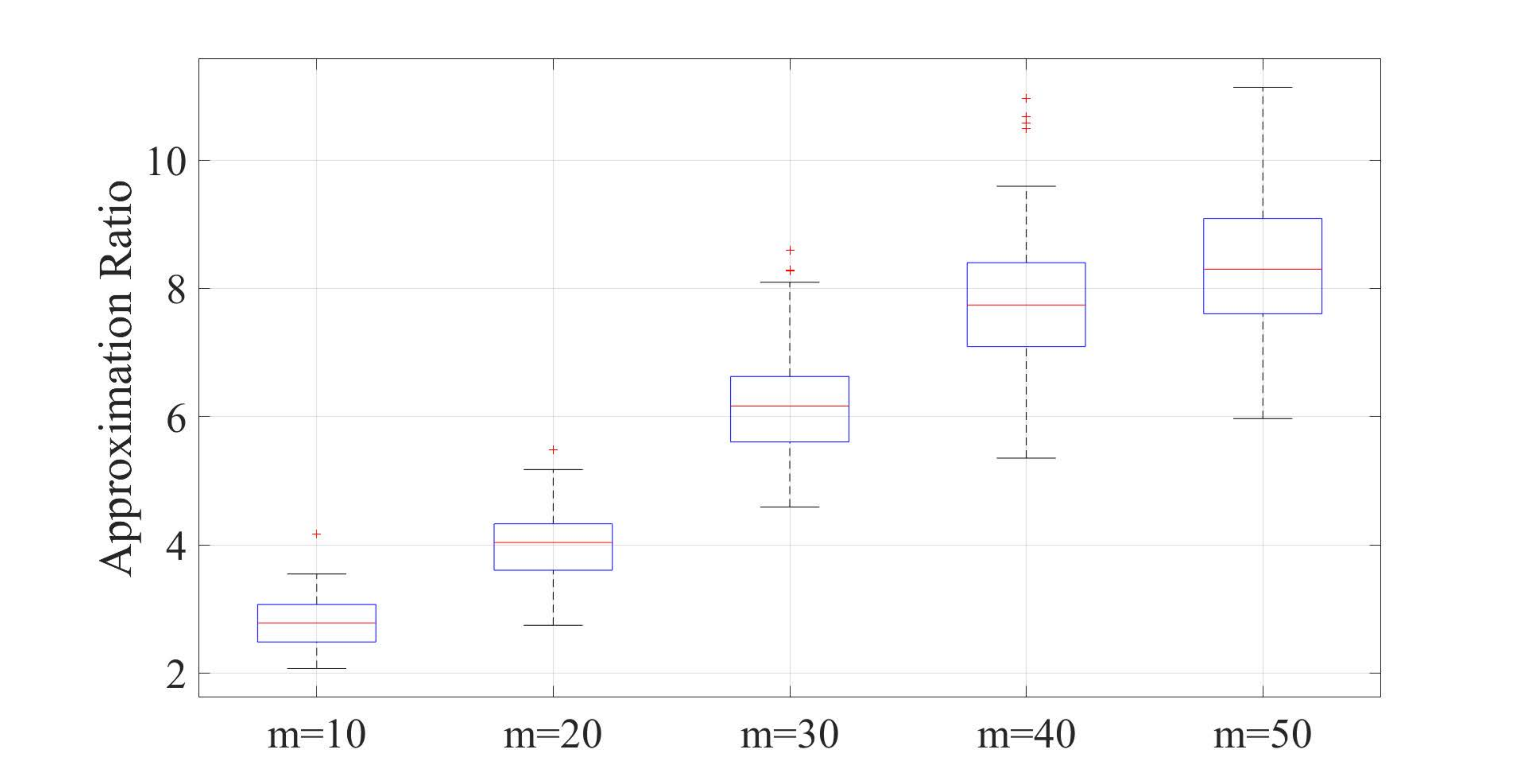}}
\end{minipage}
\label{fig:Indivisible coflows from custom box plot h}
}
\caption{The box plot of FLS, FLPT, Weaver, and CLS under synthetic traces in heterogeneous parallel networks.}
\label{fig:Coflows from custom box plot h}
\end{figure}

Figure \ref{fig:Coflows from custom box plot h} illustrates a box plot displaying the performance of FLS, FLPT, Weaver, and CLS in heterogeneous parallel networks. The synthetic trace used in this analysis involves $25$ coflows operating in a network with $m=50$ cores and $N=10$ input/output links.
For this scenario, we set $h$ as 5 and generated 100 sample traces for each algorithm. The box plot includes quartiles, maximum and minimum values for each algorithm. The results reveal that FLPT not only achieves a superior ratio compared to Weaver but also exhibits a narrower interquartile range. This is demonstrated in Figure \ref{fig:Divisible coflows from custom box plot h} and summarized in Table \ref{tableBoxPlot2}. In Figure \ref{fig:Indivisible coflows from custom box plot h}, it can be observed that the approximation ratio of CLS, as well as its interquartile range, increase as the number of cores grows. These results are consistent with those observed in identical parallel networks.

\begin{table}[!h]
\caption{The quartiles, maximum, and minimum of FLPT, Weaver, and CLS for box plots in Figure \ref{fig:Coflows from custom box plot h}.}
\centering
\renewcommand\arraystretch{1.5}
\begin{tabular}{||c|c|c|c||}
\hline
\textbf{\diagbox{Q{\&}M}{Algo}} & \textbf{FLPT} & \textbf{Weaver} & \textbf{CLS ($m=50$)}  \\ \hline
\textbf{$Q_1$} & 1.4556 & 1.5420 & 7.6049  \\ \hline
\textbf{$Q_2$} & 1.5486 & 1.6404 & 8.2986  \\ \hline
\textbf{$Q_3$} & 1.6340 & 1.8241 & 9.0896  \\ \hline
\textbf{Maximum} & 1.8936 & 2.1244 & 11.1368  \\ \hline
\textbf{Minimum} & 1.3136 & 1.3510 & 5.9693  \\ \hline
\end{tabular}
\label{tableBoxPlot2}
\end{table}

\begin{figure}[!h]
\centering
\subfigure[The performance of algorithms: FLPT and Weaver.]{
\begin{minipage}[h]{0.4\textwidth}
\centering
{\includegraphics[width=3.4in]{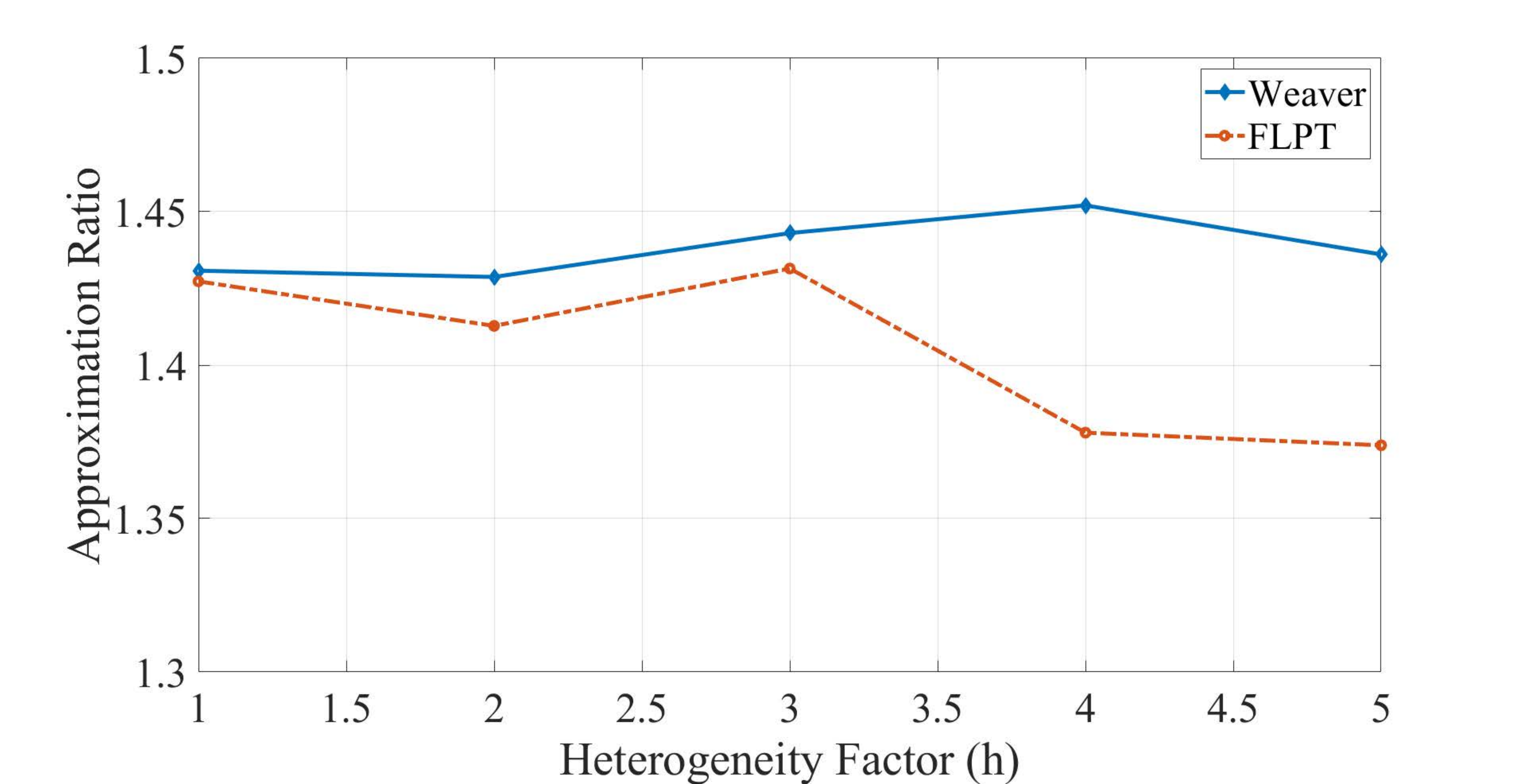}}
\end{minipage}
\label{fig:Divisible coflows from custom heterogeneous configuration 1}
}

\subfigure[The performance of algorithms: CLS.]{
\begin{minipage}[h]{0.4\textwidth}
\centering
{\includegraphics[width=3.4in]{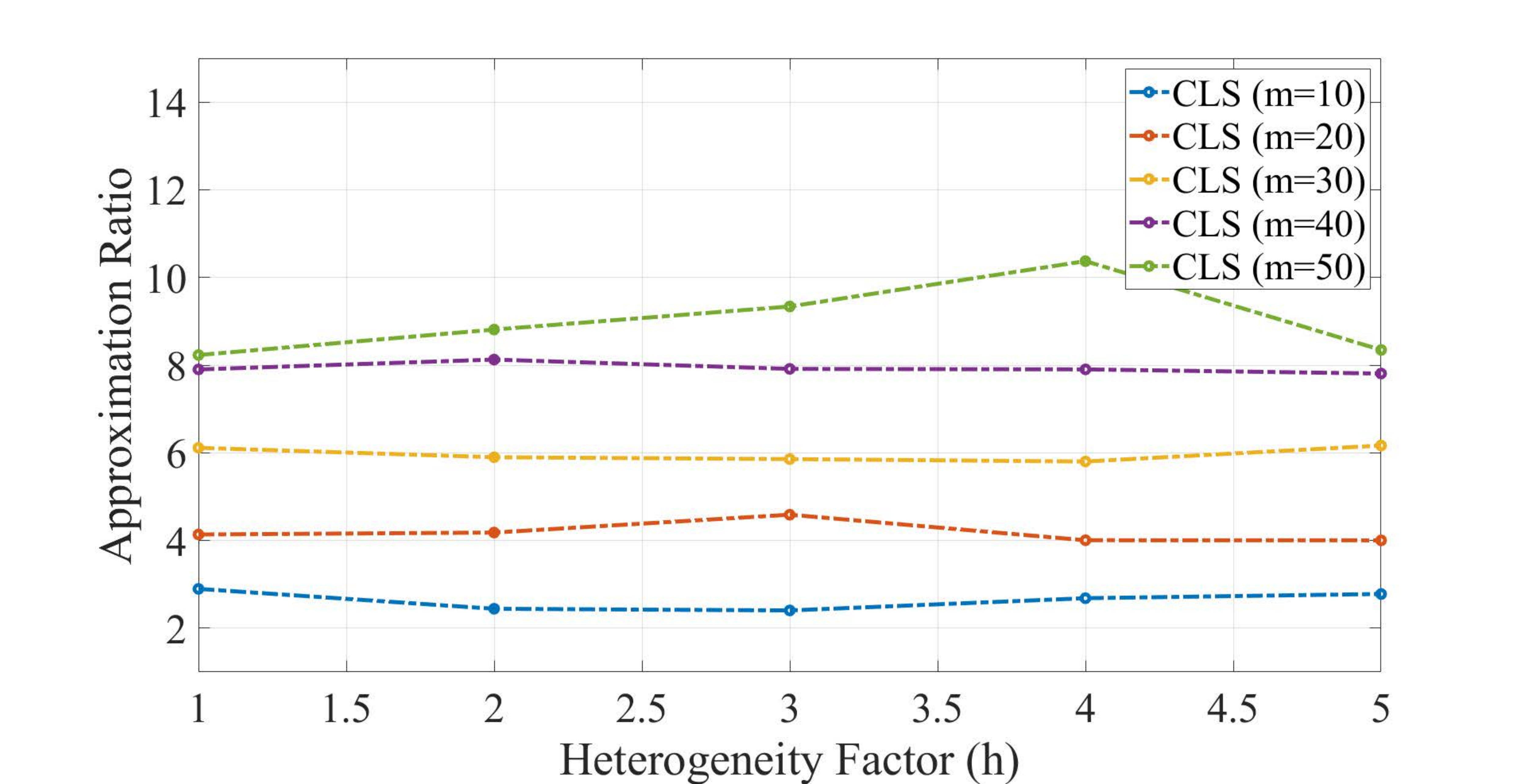}}
\end{minipage}
\label{fig:Divisible coflows from custom heterogeneous configuration 2}
}

\caption{Approximation ratio of FLPT, Weaver, and CLS for distinct heterogeneity factor under synthetic traces in heterogeneous parallel networks.}
\label{fig:Coflows from custom heterogeneous}
\end{figure}

Figure \ref{fig:Coflows from custom heterogeneous} displays the approximation ratio of FLPT, Weaver, and CLS for different heterogeneity factors in heterogeneous parallel networks. The analysis is based on a synthetic trace comprising 25 coflows and $m=20$ cores with $N=10$ input/output links in five scenarios with varying heterogeneity factors.

For each scenario, we generated 100 sample traces and reported the average performance of the algorithms. When the heterogeneity increases (with $h=1$), FLPT and Weaver exhibit similar performance. However, when the heterogeneity decreases (with $h=5$), the performance gap between FLPT and Weaver widens. Overall, FLPT outperforms Weaver in terms of performance.
Regarding CLS, there is no significant difference in performance based on heterogeneity factors. However, the number of cores has a greater impact on the performance of CLS.

\section{Conclusion}\label{section:conclusion}
This paper focuses on addressing the problem of coflow scheduling with the objective of minimizing the makespan of all network cores. We propose three algorithms that achieve approximation ratios of $3-\tfrac{2}{m}$ and $\tfrac{8}{3}-\tfrac{2}{3m}$ for the flow-level scheduling problem, and an approximation ratio of $2m$ for the coflow-level scheduling problem in identical parallel networks.

To evaluate the performance of our algorithms, we conduct experiments using both real and synthetic traffic traces, comparing them against Weaver's algorithm. Our experimental results demonstrate that our algorithms outperform Weaver's in terms of approximation ratio. Furthermore, we extend our evaluation to include heterogeneous parallel networks and find that the results align with those obtained in identical parallel networks.

As part of future work, we can explore additional constraints such as deadline constraints and consider alternative objectives such as tardiness objectives. Additionally, the problem of bandwidth allocation, which was not addressed in this paper, presents an interesting research direction. We anticipate further exploration of extended problems that arise from multiple parallel networks, as they play a crucial role in achieving quality of service.


%
\appendix
\section{Algorithms for Heterogeneous Parallel Networks}
In an extension of Algorithm \ref{FLPT}, the FLPT algorithm can be adapted to handle heterogeneous parallel networks. The modified algorithm, referred to as FLPT-h (Algorithm \ref{FLPT_h}), is outlined below. Let $s_h$ denote the speed factor of core $h$. The load of flow $(i, j, k)$ on core $h$ is given by $\tfrac{d_{i, j, k}}{s_h}$. Algorithm \ref{FLPT_h} closely resembles Algorithm \ref{FLPT}, with the only difference occurring at lines \ref{FLPT_h_update1}, \ref{FLPT_h_update2}, and \ref{FLPT_h_update3}. In line \ref{FLPT_h_update1}, we identify a core that minimizes the completion time of flow $(i, j, k)$. In lines \ref{FLPT_h_update2}-\ref{FLPT_h_update3}, the values of $load_{I}{(i,h)}$ and $load_{O}{(j,h)}$ are updated with $\tfrac{d_{i, j, k}}{s_h}$ if flow $(i,j,k)$ is assigned to core $h$.

\begin{algorithm}[!h]
 \caption{flow-longest-processing-time-first-scheduling-h} \label{FLPT_h}
 \begin{algorithmic}[1]
   \Require a set $\mathcal{F}$, which contains of all flows $(i,j,k)$, $\forall i \in \mathcal{I}, \forall j \in \mathcal{J}, \forall k \in \mathcal{K}$
   \State let $load_{I}{(i,h)}$ be the load on the $i$-th input port of the core $h$ \label{FLPT_h_init}
   \State let $load_{O}{(j,h)}$ be the load on the $j$-th output port of the core $h$
   \State let $\mathcal{A}_h$ be the set of flows allocated to the core $h$
   \State initialize both $load_{I}$ and $load_{O}$ to 0 and $\mathcal{A}_h = \emptyset$ for all $h \in \mathcal{M}$
   \For{each flow $(i, j, k) \in \mathcal{F}$ in non-increasing order of $d_{i, j, k}$, breaking ties arbitrarily} \label{FLPT_hfor1_s}
      \State $h^* = \arg\min_{h \in \mathcal{M}} \left\{load_{I}{(i,h)}+load_{O}{(j,h)}+\frac{d_{i, j, k}}{s_{h}}\right\}$ \label{FLPT_h_update1}
      \State $\mathcal{A}_{h^*} = \mathcal{A}_{h^*} \cup \{(i, j, k)\}$
      \State $load_{I}{(i,h^*)}=load_{I}{(i,h^*)} + \frac{d_{i, j, k}}{s_{h^*}}$ \label{FLPT_h_update2}
			\State $load_{O}{(j,h^*)}=load_{O}{(j,h^*)} + \frac{d_{i, j, k}}{s_{h^*}}$ \label{FLPT_h_update3}
   \EndFor \label{FLPT_hfor1_f}
   \State \textbf{return} $\left\{\mathcal{A}_h\right\}$ for $h \in \mathcal{M}$
 \end{algorithmic}
\end{algorithm}

In an extension of Algorithm \ref{CLS}, the CLS algorithm also can be adapted to handle heterogeneous parallel networks. The modified algorithm, referred to as CLS-h (Algorithm \ref{CLSH}), is outlined below. The load of coflow $k$ on core $h$ is given by $\tfrac{L_{i,k}}{s_h}$ and $\tfrac{L_{j,k}}{s_h}$ for all $i \in \mathcal{I}$and $j \in \mathcal{J}$, respectively. In line \ref{CLSH_update1}, we identify a core that minimizes the completion time of coflow $k$. In lines \ref{CLSH_update2}-\ref{CLSH_update3}, the values of $load_{I}{(i,h)}$ and $load_{O}{(j,h)}$ are updated with $\tfrac{L_{i,k}}{s_h}$ and $\tfrac{L_{j,k}}{s_h}$, respectively.

\begin{algorithm}[!h]
 \caption{coflow-list-scheduling-h} \label{CLSH}
 \begin{algorithmic}[1]
   \Require a set $\mathcal{K}$, which contains of all coflows 
   \State let $load_{I}{(i,h)}$ be the load on the $i$-th input port of the core $h$ \label{CLS_init}
   \State let $load_{O}{(j,h)}$ be the load on the $j$-th output port of the core $h$
   \State let $\mathcal{A}_h$ be the set of coflows allocated to the core $h$
   \State initialize both $load_{I}$ and $load_{O}$ to 0 and $\mathcal{A}_h = \emptyset$ for all $h \in \mathcal{M}$
   \For{each coflow $k \in \mathcal{K}$} \label{CLSfor1_s} 
      \State $h^* = \arg\min_{h \in \mathcal{M}} \max_{\forall i \in \mathcal{I}, \forall j \in \mathcal{J}} \left\{load_{I}{(i,h)}+\right.$  $\left.load_{O}{(j,h)}+\frac{L_{i,k}}{s_{h}}+\frac{L_{j,k}}{s_{h}}\right\}$ \label{CLSH_update1} 
      \State $\mathcal{A}_{h^*} = \mathcal{A}_{h^*} \cup \{k\}$
      \State $load_{I}{(i,h^*)}=load_{I}{(i,h^*)} + \frac{L_{i,k}}{s_{h^*}}$, $\forall i \in \mathcal{I}$ \label{CLSH_update2}
			\State $load_{O}{(j,h^*)}=load_{O}{(j,h^*)} + \frac{L_{j,k}}{s_{h^*}}$, $\forall j \in \mathcal{J}$ \label{CLSH_update3}
   \EndFor \label{CLSfor1_f}
   \State \textbf{return} $\left\{\mathcal{A}_h\right\}$ for $h \in \mathcal{M}$
 \end{algorithmic}
\end{algorithm}




\begin{thebibliography}{10}
\providecommand{\url}[1]{#1}
\csname url@rmstyle\endcsname
\providecommand{\newblock}{\relax}
\providecommand{\bibinfo}[2]{#2}
\providecommand\BIBentrySTDinterwordspacing{\spaceskip=0pt\relax}
\providecommand\BIBentryALTinterwordstretchfactor{4}
\providecommand\BIBentryALTinterwordspacing{\spaceskip=\fontdimen2\font plus
\BIBentryALTinterwordstretchfactor\fontdimen3\font minus
  \fontdimen4\font\relax}
\providecommand\BIBforeignlanguage[2]{{%
\expandafter\ifx\csname l@#1\endcsname\relax
\typeout{** WARNING: IEEEtran.bst: No hyphenation pattern has been}%
\typeout{** loaded for the language `#1'. Using the pattern for}%
\typeout{** the default language instead.}%
\else
\language=\csname l@#1\endcsname
\fi
#2}}

\bibitem{Ahmadi}
S.~Ahmadi, S.~Khuller, M.~Purohit, and S.~Yang, ``On scheduling coflows,''
  \emph{Algorithmica}, vol.~82, no.~12, pp. 3604--3629, 2020.


\bibitem{al2008scalable}
M.~Al-Fares, A.~Loukissas, and A.~Vahdat, ``A scalable, commodity data center
  network architecture,'' \emph{ACM SIGCOMM computer communication review},
  vol.~38, no.~4, pp. 63--74, 2008.


\bibitem{Hadoop}
D.~Borthakur, ``The hadoop distributed file system: Architecture and design,''
  \emph{Hadoop Project Website}, vol.~11, no. 2007, p.~21, 2007.


\bibitem{CYChen_het}
C.-Y. Chen, ``Scheduling coflows for minimizing the total weighted completion
  time in heterogeneous parallel networks,'' \emph{arXiv preprint
  arXiv:2204.07799}, 2022.


\bibitem{CYChen}
\BIBentryALTinterwordspacing
C.~Chen, ``Scheduling coflows for minimizing the total weighted completion time
  in identical parallel networks,'' \emph{CoRR}, vol. abs/2204.02651, 2022.
  [Online]. Available: \url{https://doi.org/10.48550/arXiv.2204.02651}
\BIBentrySTDinterwordspacing


\bibitem{CYChen_prec}
C.-Y. Chen, ``Scheduling coflows with precedence constraints for minimizing the
  total weighted completion time in identical parallel networks,'' \emph{arXiv
  preprint arXiv:2205.02474}, 2022.


\bibitem{website:benchmark}
M.~Chowdhury, ``Coflow-benchmark,'' https://github.com/coflow/coflow-benchmark.


\bibitem{website:mosharaf}
------, ``Coflowsim,'' https://github.com/coflow/coflowsim.


\bibitem{Coflow_cluster}
M.~Chowdhury and I.~Stoica, ``Coflow: A networking abstraction for cluster
  applications,'' in \emph{Proceedings of the 11th ACM Workshop on Hot Topics
  in Networks}, 2012, pp. 31--36.


\bibitem{Chowdhury}
------, ``Efficient coflow scheduling without prior knowledge,'' \emph{ACM
  SIGCOMM Computer Communication Review}, vol.~45, no.~4, pp. 393--406, 2015.


\bibitem{Varys}
M.~Chowdhury, Y.~Zhong, and I.~Stoica, ``Efficient coflow scheduling with
  varys,'' in \emph{Proceedings of the 2014 ACM conference on SIGCOMM}, 2014,
  pp. 443--454.


\bibitem{MapReduce}
J.~Dean and S.~Ghemawat, ``Mapreduce: simplified data processing on large
  clusters,'' \emph{Communications of the ACM}, vol.~51, no.~1, pp. 107--113,
  2008.


\bibitem{Graham69}
\BIBentryALTinterwordspacing
R.~L. Graham, ``Bounds on multiprocessing timing anomalies,'' \emph{SIAM
  Journal on Applied Mathematics}, vol.~17, no.~2, pp. 416--429, 1969.
  [Online]. Available: \url{https://doi.org/10.1137/0117039}
\BIBentrySTDinterwordspacing


\bibitem{greenberg2009vl2}
A.~Greenberg, J.~R. Hamilton, N.~Jain, S.~Kandula, C.~Kim, P.~Lahiri, D.~A.
  Maltz, P.~Patel, and S.~Sengupta, ``Vl2: A scalable and flexible data center
  network,'' in \emph{Proceedings of the ACM SIGCOMM 2009 conference on Data
  communication}, 2009, pp. 51--62.


\bibitem{Hasnain}
A.~Hasnain and H.~Karl, ``Coflow scheduling with performance guarantees for
  data center applications,'' in \emph{2020 20th IEEE/ACM International
  Symposium on Cluster, Cloud and Internet Computing (CCGRID)}.\hskip 1em plus
  0.5em minus 0.4em\relax IEEE, 2020, pp. 850--856.


\bibitem{Weaver}
X.~S. Huang, Y.~Xia, and T.~E. Ng, ``Weaver: Efficient coflow scheduling in
  heterogeneous parallel networks,'' in \emph{2020 IEEE International Parallel
  and Distributed Processing Symposium (IPDPS)}.\hskip 1em plus 0.5em minus
  0.4em\relax IEEE, 2020, pp. 1071--1081.


\bibitem{Dryad}
M.~Isard, M.~Budiu, Y.~Yu, A.~Birrell, and D.~Fetterly, ``Dryad: distributed
  data-parallel programs from sequential building blocks,'' in
  \emph{Proceedings of the 2nd ACM SIGOPS/EuroSys European Conference on
  Computer Systems 2007}, 2007, pp. 59--72.


\bibitem{Qiu}
Z.~Qiu, C.~Stein, and Y.~Zhong, ``Minimizing the total weighted completion time
  of coflows in datacenter networks,'' in \emph{Proceedings of the 27th ACM
  symposium on Parallelism in Algorithms and Architectures}, 2015, pp.
  294--303.


\bibitem{Sachdeva}
S.~Sachdeva and R.~Saket, ``Optimal inapproximability for scheduling problems
  via structural hardness for hypergraph vertex cover,'' in \emph{2013 IEEE
  Conference on Computational Complexity}.\hskip 1em plus 0.5em minus
  0.4em\relax IEEE, 2013, pp. 219--229.


\bibitem{Shafiee2}
M.~Shafiee and J.~Ghaderi, ``Scheduling coflows in datacenter networks:
  Improved bound for total weighted completion time,'' \emph{ACM SIGMETRICS
  Performance Evaluation Review}, vol.~45, no.~1, pp. 29--30, 2017.


\bibitem{Shafiee}
------, ``An improved bound for minimizing the total weighted completion time
  of coflows in datacenters,'' \emph{IEEE/ACM Transactions on Networking},
  vol.~26, no.~4, pp. 1674--1687, 2018.


\bibitem{Shen}
D.~Shen, J.~Luo, F.~Dong, and J.~Zhang, ``Virtco: joint coflow scheduling and
  virtual machine placement in cloud data centers,'' \emph{Tsinghua Science and
  Technology}, vol.~24, no.~5, pp. 630--644, 2019.


\bibitem{singh2015jupiter}
A.~Singh, J.~Ong, A.~Agarwal, G.~Anderson, A.~Armistead, R.~Bannon, S.~Boving,
  G.~Desai, B.~Felderman, P.~Germano, \emph{et~al.}, ``Jupiter rising: A decade
  of clos topologies and centralized control in google's datacenter network,''
  \emph{ACM SIGCOMM computer communication review}, vol.~45, no.~4, pp.
  183--197, 2015.


\bibitem{williamson_shmoys_2011}
D.~P. Williamson and D.~B. Shmoys, \emph{Greedy Algorithms and Local
  Search}.\hskip 1em plus 0.5em minus 0.4em\relax Cambridge University Press,
  2011, p. 27–56.


\end{thebibliography}
%

\end{document}